\documentclass[authoryear,review,11pt,a4paper,nonatbib]{elsarticle} 

\usepackage{amssymb}
\usepackage{amsthm}

\journal{European Journal of Operational Research}

\usepackage{url}
\usepackage[colorlinks,
linkcolor=blue,
anchorcolor=blue,
citecolor=blue,
filecolor=blue,
linkcolor=blue,
runcolor=blue,
urlcolor=blue,
]{hyperref}

\usepackage{amsmath}
\usepackage{booktabs,graphics,mathrsfs,graphicx,amsfonts,balance,color}
\usepackage[misc]{ifsym}
\usepackage{enumerate}

\usepackage[round,comma,authoryear]{natbib}  
\usepackage{apalike}
\usepackage{arydshln} 
\usepackage[left=1.8cm, right=1.8cm, top=3cm, bottom=2cm]{geometry}

\usepackage[dvipsnames, svgnames, x11names]{xcolor} 

\usepackage{multirow}

\usepackage{threeparttable} 

\usepackage{amssymb}
\usepackage{subfig}
\usepackage{bm}
\usepackage{nomencl}
\usepackage[shortlabels]{enumitem}
\usepackage{amsthm}
\usepackage{ragged2e} 
\usepackage{algpseudocode}
\usepackage{booktabs}
\usepackage{graphicx}
\usepackage{placeins}
\usepackage{algorithm}
\usepackage{algorithmicx}

\newtheorem{assumption}{Assumption}
\newtheorem{definition}{Definition}

\newtheorem{proposition}{Proposition}

\begin{document}
	
	\begin{frontmatter}
		
		
		
		\title{Covariance Matrix Estimation for Positively Correlated Assets}
		
	\author[1]{Weilong Liu} 
	\ead{liuwlong7@mail.sysu.edu.cn}
	
	\author[1]{Yanchu Liu\corref{cor1}}
	\ead{liuych26@mail.sysu.edu.cn}

	\cortext[cor1]{Corresponding author.}
	
	\affiliation[1]{organization={Lingnan College, Sun Yat-sen University}, 
		addressline={},
		city={Guangzhou},
		postcode={510275},
		state={},
		country={China}}

		\begin{abstract}
		The comovement phenomenon in financial markets creates decision scenarios with positively correlated asset returns. This paper addresses covariance matrix estimation under such conditions, motivated by observations of significant positive correlations in factor-sorted portfolio monthly returns. We demonstrate that fine-tuning eigenvectors linked to weak factors within rotation-equivariant frameworks produces well-conditioned covariance matrix estimates. Our Eigenvector Rotation Shrinkage Estimator (ERSE) pairwise rotates eigenvectors while preserving orthogonality, equivalent to performing multiple linear shrinkage on two distinct eigenvalues. Empirical results on factor-sorted portfolios from the Ken French data library demonstrate that ERSE outperforms existing rotation-equivariant estimators in reducing out-of-sample portfolio variance, achieving average risk reductions of 10.52\% versus linear shrinkage methods and 12.46\% versus nonlinear shrinkage methods. Further checks indicate that ERSE yields covariance matrices with lower condition numbers, produces more concentrated and stable portfolio weights, and provides consistent improvements across different subperiods and estimation windows.
		\end{abstract}
		
		
		
		
	\end{frontmatter}
	
	\section{Introduction}
	The widespread phenomenon of comovement in financial markets, where the returns of multiple assets tend to move in the same direction, gives rise to various decision scenarios with positively correlated assets. This phenomenon typically occurs during market-wide events, such as economic shifts or global crises, as common factors influence all candidate assets. Researchers identify significant common factors in the returns of small-cap stocks, value stocks, closed-end funds, industry-specific stocks, and bonds with the same rating and maturity, revealing numerous patterns of comovement in asset returns \citep{BARBERIS2005}. Subsequent studies also find strong comovement signals in various contexts, such as stocks of firms headquartered in the same geographic area \citep{PIRINSKY2006}, stocks with similar prices \citep{Green2009}, and stocks within heavily regulated industries \citep{BLAU2023}. Researchers also demonstrate that high-frequency trading significantly enhances the comovement phenomenon among assets \citep{MALCENIECE2019}. As a prominent example of computational advances, high-frequency trading is expected to become increasingly prevalent, making decision scenarios characterized by positively correlated assets more common. This paper addresses the challenge of covariance matrix estimation for such scenarios.
	
	An additional motivation for this study arises from our observation of persistently high positive correlations among assets within factor-sorted portfolios. As shown in Section \ref{PosCor}, we test 19 monthly datasets from the Ken-French website, which are based on diversified factors and encompass the most commonly studied factors available on the site. Our empirical analysis reveals that the average sample correlation coefficient significantly exceeds 0.5, while even the lowest pairwise correlation is positive, with these results consistently holding over time. Based on this observation, we introduce a positive correlation assumption for such decision scenarios: all elements of the sample correlation matrix are positive. This assumption may seem overly stringent, as it does not imply an absolute positive correlation between any pair of individual stocks. However, it holds in specific decision scenarios due to the comovement phenomenon, such as in the factor-sorted portfolios discussed above. This paper is developed to propose a customized covariance matrix estimator under the positive correlation assumption.
	
	
	Covariance matrix estimation is crucial for financial decision-making tasks such as asset pricing, risk management, and portfolio optimization. Jensen's Inequality suggests that the largest eigenvalue of the sample covariance matrix is biased upwards, while the smallest eigenvalue is biased downwards, indicating that the sample covariance matrix is prone to being ill-conditioned \citep{Vaart1961, Ledoit2004}. In Principal Component Analysis (PCA), the eigenvectors denote the principal components, which capture the directions that account for the maximum variance in the data, while the eigenvalues quantify the variance each component explains. \cite{Fan2013} introduces the pervasiveness (strong factor) assumption, which posits that the first $K$ largest eigenvalues of the data covariance matrix diverge while the remaining eigenvalues remain bounded. This results in a pronounced gap between the $K$-th and the ($K+1$)-th largest eigenvalues, a property known as the strong factor structure. Several subsequent studies adopt this strong factor assumption \citep{DAI2019,KIM2019,Ding2021}. \cite{Dai2024} further demonstrates that the sample covariance matrix of S\&P 500 stocks exhibits a single dominant factor, with the largest eigenvalue significantly outweighing the others. The financial interpretation of the rotation-equivariant method is to repackage the $n$ individual assets into an equal number of funds that span the space of attainable investment opportunities \citep{Ledoit2017}. The eigenvalues represent the variances of the returns of these $n$ funds, and underestimating the small eigenvalues close to zero can lead to a significant underestimation of the associated risk.
	
	We reconsider the above issue for positively correlated assets, focusing on the correlation matrix as an alternative to the covariance matrix in the normalized setting. Consistent with \cite{Dai2024}, our results show that the sample correlation matrix has a dominant eigenvalue that significantly influences the spectrum, while the remaining eigenvalue is much smaller. Additionally, we observe a misalignment in the distribution of the sample eigenvectors. According to the Perron-Frobenius theorem, we find that the eigenvector corresponding to the largest eigenvalue (i.e., the dominant eigenvector) tends to align closely with a uniform vector, while the remaining eigenvectors, orthogonal to it, approach the null space of the uniform vector. Our theoretical findings indicate that the deviation of the dominant eigenvector from the null space is significantly larger than that of the other eigenvectors, with this discrepancy being more pronounced than the differences in the corresponding eigenvalues. Moreover, we observe that the eigenvectors associated with weak factors exhibit minor deviations, and moderately increasing these deviations in the rotation-equivariant estimation framework can push up the corresponding eigenvalues.
	
	A consensus view holds that reducing the dispersion among the sample eigenvalues can provide a well-conditioned covariance matrix estimator in high-dimensional asset settings \citep{Ledoit2012,Shi2020,Bodnar2018}. However, determining the appropriate shrinkage coefficients remains challenging, and various methods have been developed to address this challenge \citep{Ledoit2004, Ledoit2012, Barroso2022}. Within the rotation-equivariant framework, \cite{Ledoit2017, Ledoit2022} provide an oracle estimator that relies on the true covariance matrix \(\Sigma\), where the estimated eigenvalue is constructed as $\hat{\lambda}_{i} = \boldsymbol{q}_{i}^T \Sigma \boldsymbol{q}_{i}$ with its corresponding sample eigenvector \(\boldsymbol{q}_{i}\). In practical applications, the unobservable true covariance matrix \(\Sigma\) is replaced by an available target matrix \citep{Ledoit2004a, Ledoit2003, Touloumis2015}. Unlike the aforementioned estimators, which focus on target matrices, this paper corrects the eigenvector $\boldsymbol{q}_{i}$ while using the sample covariance matrix as the target matrix. Our correction of the eigenvectors is also motivated by the conclusion in \cite{DeMiguel2009}. The financial interpretation of the eigenvector is the portfolio vector of the corresponding repackaged funds. We prove that a vector's deviation from the uniform vector's null space is negatively correlated with the \( \ell_2 \)-norm of its corresponding unit-cost portfolio. As discussed in \cite{DeMiguel2009}, introducing the \( \ell_2 \)-norm constraint helps mitigate estimation errors in the covariance matrix. Therefore, a natural approach is to develop a shrinkage estimator for the sample covariance matrix by correcting the misalignment of its eigenvectors.
	
	We propose a new rotation-equivariant covariance matrix estimator, the Eigenvector Rotation Shrinkage Estimator (ERSE), designed for positively correlated assets. ERSE constrains each eigenvector’s deviation from the null space of the uniform vector to exceed a certain threshold. To reduce the deviation between paired eigenvectors while preserving their orthogonality, we introduce the Paired Eigenvector Rotation (PER) technique. The ERSE method is implemented by iteratively applying the PER technique until all eigenvectors satisfy the deviation requirements. Typically, sample eigenvectors associated with weak factors exhibit minor deviations, prompting the ERSE method to apply rotation to these eigenvectors. This rotation process generally amplifies weak eigenvalues while diminishing the corresponding strong eigenvalues. We demonstrate that the ERSE method is equivalent to repeatedly applying linear shrinkage to the two differentiated eigenvalues, narrowing their gap while maintaining their sum. Furthermore, the computational complexity of ERSE is manageable, as the PER technique can be executed in a few steps and requires no more than $n-1$ iterations. Experimental results on factor-sorted portfolios, in which each pair of assets exhibits high positive correlation, demonstrate that the ERSE method outperforms alternative estimators in reducing out-of-sample risk. Additional robustness analyses, including subperiod evaluations, varying estimation windows, and random asset sampling, further confirm the advantages of the ERSE method.
	
	This paper contributes to the literature on Global Minimum-Variance (GMV) portfolio models. The seminal work of \cite{Markowitz1952} remains the cornerstone of both theoretical research and practical applications in portfolio management. Existing literature highlights that estimation errors in the mean vector tend to have a significantly more detrimental effect on out-of-sample portfolio performance than errors in the covariance matrix, leading many scholars to advocate for the GMV portfolio model, which focuses solely on minimizing risk without relying on expected returns. Empirical evidence from \cite{DeMiguel2009} shows that the GMV portfolio often exhibits superior out-of-sample Sharpe ratios and lower turnover rates than other mean-variance models. \cite{Fan2012} suggests considering the expected return of all assets as zero, positioning the GMV portfolio as the sole efficient portfolio within the mean-variance framework. \cite{Cai2020} discusses the challenges posed by microstructural noise in high-frequency asset returns and advocates for a GMV-based approach to constructing trading strategies. Recent studies also suggest directly predicting the optimal weights of the GMV portfolio \citep{Ding2021,Bodnar2017,Reh2023}. This paper addresses the problem of estimating the covariance matrix for positively correlated assets, thereby enhancing decision support for constructing GMV portfolios in such environments.
	
	Our work is related to the literature on factor-based covariance estimation methods. \cite{Conlon2025} conduct a comprehensive comparative analysis of a wide range of factor models and the factor-based framework for constructing the covariance matrix can enhance minimum-variance portfolio performance. \cite{Fan2008} propose covariance estimators with observable factors, while \cite{Fan2013} extends this framework to unobservable factors and introduces the well-known Principal Orthogonal Complement Thresholding (POET) estimator. Several modifications of POET have been proposed. For example, \cite{Fan2018} adapts the method for heavy-tailed data and elliptical factor models, \cite{Wang2021} extends POET to a kernel-based local estimation technique for large dynamic covariance matrices using index variables, and \cite{Dai2024} extends POET to estimate large covariance matrices with a mixed structure of observable and unobservable strong and weak factors. Broadly speaking, the POET method performs PCA on the sample covariance matrix, extracting the top $K$ principal components as the signal and applying thresholding to the remaining noise, resulting in a sparse idiosyncratic covariance estimator. POET and its extensions improve covariance matrix estimators by correcting errors associated with weak factors. Similarly, our proposed ERSE method focuses on estimation errors associated with weak factors. For positively correlated assets, we find that the eigenvectors corresponding to weak factors converge to the null space of uniform vectors. By adjusting the deviation of these eigenvectors, we can effectively reduce the estimation error of their corresponding eigenvalues.
	
	Our estimation method is based on the rotation-equivariant estimation framework and contributes to the literature on shrinkage estimation of the covariance matrix. \cite{Stein1986} demonstrate that superior covariance estimators can be constructed by shrinking the eigenvalues of the sample covariance matrix. Ledoit and Wolf \cite{Ledoit2004} introduces a promising shrinkage estimator formed by a convex combination of the sample covariance matrix and the identity matrix scaled by the mean of the sample eigenvalues. Alternatives to the identity matrix as a shrinkage target include the constant correlation model \citep{Ledoit2004a}, the single index model \citep{Ledoit2003}, and the diagonal matrix of sample eigenvalues \citep{Touloumis2015}. These linear shrinkage techniques offer flexible and practical solutions for covariance matrix estimation. A natural extension of \cite{Ledoit2012} is the introduction of a nonlinear shrinkage estimator that applies individual shrinkage intensities to each sample eigenvalue, demonstrating excellent out-of-sample performance in portfolio selection \citep{Ledoit2017}. \cite{Ledoit2020} develops the first analytical formula for nonlinear shrinkage estimation of large-dimensional covariance matrices. Other nonlinear estimation methods have been explored in various decision-making contexts \citep{Won2013, Abadir2014, Lam2016, Nguyen2022,Morstedt2024}. \cite{Barroso2022} proposes determining shrinkage coefficients for eigenvalues and other parameters based on out-of-sample errors. The proposed ERSE method also performs shrinkage estimation of the covariance matrix by enlarging small eigenvalues and shrinking large ones. We demonstrate that the ERSE method is equivalent to applying shrinkage corrections to some or all eigenvalues, where the deviation of each eigenvector determines the shrinkage coefficient for its corresponding eigenvalue.
	
	The remainder of this paper is organized as follows: Section \ref{PosCor} examines the positive correlations among factor-sorted portfolios from Ken-French's website. Section \ref{de} investigates the misalignment of sample eigenvectors. Section \ref{model} proposes a novel covariance matrix estimator designed for positively correlated assets. Section \ref{rela} discusses the relationship between the proposed estimator and existing methods. Section \ref{res} presents empirical results, and Section \ref{conclu} concludes the paper.
	
	\section{Positive Correlations in Factor-Sorted Portfolios}
	\label{PosCor}
	In this section, we examine the correlations among factor-sorted portfolios derived from the well-known Ken-French database\footnote{Ken-French website: http://mba.tuck.dartmouth.edu/pages/faculty/ken.french/data\_library.html}. The database ranks companies based on various financial characteristics, such as market capitalization, book-to-market ratio, and momentum, and groups them using percentile breakpoints. Each stock group is treated as a tradable asset, with returns represented by the value-weighted returns of its constituent assets. 
	
	We adopt 19 datasets of factor-sorted portfolios, including two industry-level portfolios: the 30-industry (30IN) and 49-industry (49IN) portfolios representing the U.S. stock market; fourteen 5$\times$5 double-sorted portfolios formed on size and book-to-market (25SB), size and operating profitability (25SO), size and investment (25SI), book-to-market and operating profitability (25BO), book-to-market and investment (25BI), operating profitability and investment (25OI), size and momentum (25SM), size and short-term reversal (25SS), size and long-term reversal (25SL), size and accruals (25SA), size and market beta (25SBe), size and net share issues (25SN)\footnote{On the Ken-French website, 25SN is structured as a 5$\times$7 dataset with 35 instead of 25 portfolios.}, size and variance (25SV), and size and residual variance (25SR); three 10$\times$10 portfolios double-sorted portfolios formed on  size and book-to-market (100SB), size and operating profitability (100SO), and size and investment (100SI). 
	
	
	We reconstruct 10 datasets with varying numbers of underlying assets. The first five datasets directly use the original data from 30IN, 49IN, 100SB, 100SO, and 100SI. Since each dataset on the Ken-French website contains no more than 100 assets, we adopt the method outlined in \cite{Shi2020} to combine multiple datasets and create larger test datasets. The remaining five datasets, denoted as MD datasets, are synthesized by combining subsets of the original datasets. Specifically, 150MD combines six 5$\times$5 datasets (25SB, 25SO, 25SI, 25BO, 25BI, and 25OI), 200MD integrates eight 5$\times$5 datasets (25SM, 25SS, 25SL, 25SA, 25SBe, 25SN, 25SV, and 25SR), and 300MD is constructed from the three 10$\times$10 datasets (100SB, 100SO, and 100SI). To expand further, 500MD combines 300MD and 200MD, and the final dataset, 650MD, combines 500MD and 150MD. These datasets are constructed using monthly frequency data from July 1969 to June 2024, covering 55 years (660 months). We exclude assets with more than 10 missing observations, and missing values within the remaining assets are replaced with zeros to maintain a balanced panel for analysis. Table \ref{ds} summarizes the key characteristics of the datasets.
	
	\begin{table}[!htb]\small
		\centering
		\caption{Summary statistics of the datasets}
		\resizebox{\textwidth}{!}{
			\begin{tabular}{lcccccccccc}
				\toprule
				& 30IN & 49IN & 100SB & 100SI & 100SO & 150MD & 200MD & 300MD & 500MD & 650MD \\
				\midrule
				\# of Assets & 30 & 49 & 96 & 99 & 99 & 150 & 210 & 294 & 504 & 654 \\
				\# of Months & 660 & 660 & 660 & 660 & 660 & 660 & 660 & 660 & 660 & 660 \\
				Average Mean & 1.0042 & 0.9966 & 1.0756 & 1.0991 & 1.0724 & 1.0488 & 1.0609 & 1.0824 & 1.0734 & 1.0678 \\
				Average Variance & 41.399 & 46.720 & 40.004 & 37.472 & 38.513 & 32.233 & 35.847 & 38.650 & 37.482 & 36.278 \\
				Maximum Correlation & 0.8589 & 0.8666 & 0.9516 & 0.9398 & 0.9394 & 0.9797 & 0.9972 & 0.9741 & 0.9972 & 0.9972 \\
				Average Correlation & 0.5983 & 0.5644 & 0.7437 & 0.7679 & 0.7642 & 0.7909 & 0.8251 & 0.7594 & 0.7862 & 0.7855 \\
				Minimum Correlation & 0.2478 & 0.0717 & 0.4081 & 0.5017 & 0.3598 & 0.4511 & 0.4163 & 0.3598 & 0.3411 & 0.3411 \\
				\bottomrule
			\end{tabular}
		}
		\label{ds}
	\end{table}
	
	Table \ref{ds} highlights a significant characteristic of factor-sorted portfolios in the U.S. stock market: the strong positive correlations among asset returns. Two key observations emphasize this comovement phenomenon. First, the average pairwise correlation is consistently high, with correlation coefficients exceeding 0.5 across all datasets. The minimum pairwise correlation coefficient remains positive for all datasets, with most datasets showing values well above 0.2. The 49IN dataset has the lowest minimum correlation at 0.0717, which is still positive. This finding indicates that none of the datasets include assets with negative correlation coefficients, further reinforcing the overall trend of positive correlations between asset returns in the factor-sorted portfolios.
	
	Using a sliding window approach, we further examine the dynamic evolution of correlation coefficients among the factor-sorted portfolios. Specifically, we use a window length of 120 months, a common choice in financial research \citep{DeMiguel2009, Shi2020}. Starting from month 121, we calculate the correlation coefficients using the most recent 120 months of return data for each time step. Our analysis focuses on two key measures of correlation: the average and minimum pairwise correlations between assets. Figure \ref{cmin} shows the time evolution of these two correlation measures across different datasets. The results indicate that the minimum correlation between assets consistently remains above zero across all datasets, while the average correlation stays relatively high throughout the observed period. These findings suggest a persistent positive correlation among assets across all datasets, supporting the hypothesis that factor-sorted portfolios exhibit a strong positive correlation over time.
	\begin{figure}[!htb]
		\centering
		\includegraphics[width=1.05\linewidth]{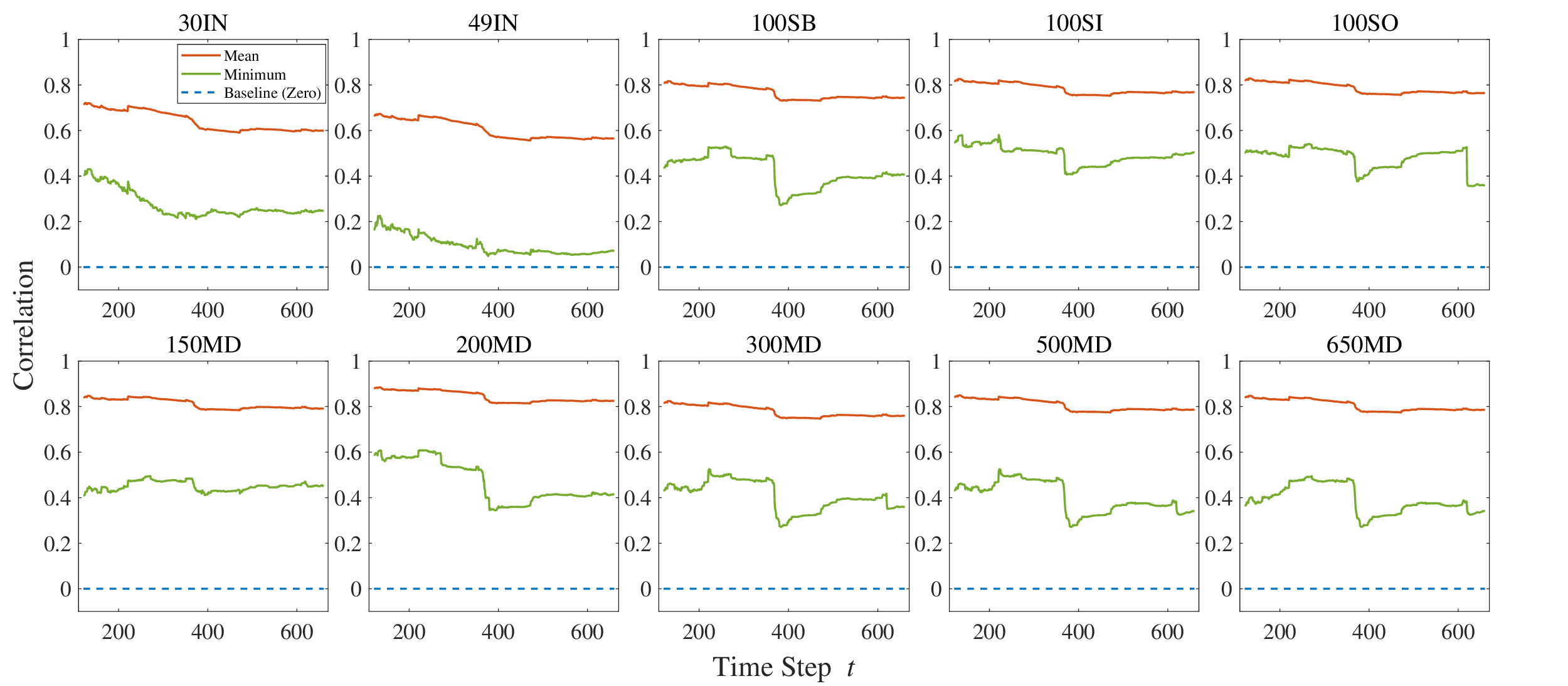}
		\caption{Dynamic Evolution of the Average and Minimum Correlation coefficients}
		\label{cmin}
	\end{figure}

	In conclusion, the analysis of correlation coefficients reveals strong and persistent positive correlations among factor-sorted portfolios. The findings highlight a decision-making environment in financial markets characterized by positively correlated assets, which serves as the foundation and motivation for this study. As a key statistical tool for measuring risk, the covariance matrix plays a crucial role in asset pricing and portfolio optimization. Building on this insight, we propose a covariance matrix estimation method specifically designed for positively correlated assets, with the aim of providing more effective risk management tools for such decision-making scenarios.
	
	\section{Deviation of Eigenvectors from the Null Space of the Uniform Vector}
	\label{de}
	The number of assets in the investable universe is denoted as \( n \). We assume that the return vector for these \( n \) assets follows a multivariate normal distribution with the covariance matrix \( \Sigma \). At the start of the decision period, we have access to the most recent \( L \) periods of historical return data, represented as \( \boldsymbol{r}_1, \ldots, \boldsymbol{r}_L \), which serve as the basis for making investment decisions. The sample covariance matrix, which acts as the maximum likelihood estimator, is given by:
	\begin{align}
		\Sigma^{S} = \frac{1}{L-1}\cdot\sum_{t=1}^{L}(\boldsymbol{r}_{t}-{\boldsymbol{m}})(\boldsymbol{r}_{t}-{\boldsymbol{m}})^{T}
	\end{align}
	where $\boldsymbol{m}$ is the sample mean vector of the returns $\boldsymbol{r}_{1},\ldots,\boldsymbol{r}_{L}$. Then, the sample correlation matrix, denoted by $R^{S}$, can be derived from the covariance matrix $\Sigma^{S}$ as:
	\begin{align}
		R^{S} = (D^{S})^{-1} \Sigma^{S} (D^{S})^{-1}
	\end{align}
	where $D^{S}={\rm Diag}(\sqrt{\Sigma^{S}_{11}},\ldots,\sqrt{\Sigma^{S}_{nn}})$ is a diagonal matrix with diagonal elements representing the sample standard deviations. 
	
	As discussed earlier, we focus on the problem of covariance matrix estimation for positively correlated assets. In this context, we assume that the historical returns of any pair of assets exhibit a positive correlation, as outlined in Assumption \ref{ass_pos}.
	\begin{assumption}
		\label{ass_pos}
		The sample correlation matrix for the \( n \) assets is denoted by \( R^S = (\rho_{i,j})_{n \times n} \). We assume that \( R^S \) is a positive matrix, with each sample correlation coefficient \( \rho_{i,j} > 0 \), indicating that all asset pairs exhibit a positive correlation.
	\end{assumption}
	
	The assumption stated above generally does not apply to individual stocks, where correlations between two assets can be either positive or negative. However, this assumption holds in certain decision-making environments due to the comovement phenomenon in financial markets. For example, as discussed in Section \ref{PosCor}, empirical evidence suggests that sample covariance matrices for factor-sorted portfolios are typically positive, indicating that the assumption holds when assets are grouped based on common factors. Consequently, investors can easily construct positively correlated portfolios using factor-sorted portfolios. Additionally, this assumption may hold in environments driven by information-based comovement or in scenarios where assets are clustered based on their correlations.
	
	In this paper, we construct the covariance matrix estimator by applying the rotation-equivariant estimation method to the correlation matrix, as outlined in Assumption \ref{ass_c}.
	\begin{assumption}
		\label{ass_c}
		Denote the spectral decomposition of the sample correlation matrix by \( R^{S} = Q^{S} \Lambda^{S} (Q^{S})^{T} \), where \( \Lambda^{S} = {\rm Diag}(\lambda_{1}^{S}, \ldots, \lambda_{n}^{S}) \) is a diagonal matrix containing the \( n \) sample eigenvalues, and \( Q^{S} \) is the matrix whose columns, \( \boldsymbol{q}_{i}^{S} \), are the eigenvectors corresponding to the eigenvalue \( \lambda_{i}^{S} \). We consider a class of covariance matrix estimators of the following form: 
		\begin{align}
			\hat{\Sigma} = D^{S} Q^{S} {\rm Diag}(\hat{\boldsymbol{\lambda}}) (Q^{S})^{T} D^{S}
		\end{align}
		where \( \hat{\boldsymbol{\lambda}} = (\hat{\lambda}_{1}, \ldots, \hat{\lambda}_{n})^{T} \) is an estimator for the vector of eigenvalues.
	\end{assumption}
	
	Existing rotation-equivariant estimators typically focus on shrinking the eigenvalues of the sample covariance matrix \citep{Ledoit2017, Ledoit2022}. In contrast, this study takes a different approach by shrinking the eigenvalues of the correlation matrix. The correlation matrix captures the linear relationships between variables while removing scale effects, and adjusting its eigenvalues may provide a more stable and accurate representation of the underlying dependencies. \cite{Barroso2022} demonstrates the advantages of using the correlation matrix as the estimation target, particularly in enhancing out-of-sample performance in portfolio optimization. It is important to note that we do not claim that using the correlation matrix is inherently superior to using the covariance matrix directly. Moreover, we do not impose the constraint that the diagonal elements of the estimated correlation matrix must equal one. Deviations from one reflect slight adjustments to the sample variances of asset returns.
	
	We revise the sample eigenvalues of the correlation matrix while keeping the sample eigenvectors. The main reason is that improving upon the sample eigenvectors is difficult, as pointed out by \cite{Stein1986} and \cite{Shi2020}. Consequently, the problem of estimating the covariance matrix can be reformulated as a multivariate prediction task focused on the eigenvalues $\hat{\lambda}_{t,1},\ldots,\hat{\lambda}_{t,n}$. It is straightforward to prove that, as long as the $n$ eigenvalue estimators of the correlation matrix are strictly positive, we can reconstruct a positive definite estimator of the covariance matrix based on Assumption \ref{ass_c}.
	
	It is well known that in high-dimensional asset environments, the eigenvalues of the sample covariance (or correlation) matrix often exhibit excessive dispersion, with larger eigenvalues frequently overestimated and smaller eigenvalues typically underestimated \citep{Stein1986, Ledoit2017, Shi2020, Nguyen2022}. The significant underestimation of small eigenvalues often leads to severe instability in the sample covariance matrix. An interesting phenomenon observed in this study is that, for portfolios of positively correlated assets, the eigenvectors associated with these small eigenvalues tend to converge to the null space of the uniform vector, providing important support for the shrinkage estimation method proposed in this paper. To better explain this phenomenon, we first provide the following definition:
	\begin{definition}
		\label{Tq}
		The null space of the uniform vector is the set of all vectors that are orthogonal to the uniform vector, which is formulated as:
		\begin{align}
			N = \left\{ \boldsymbol{x} \in \mathbb{R}^n : \boldsymbol{1}_n^T \boldsymbol{x} = 0 \right\}
		\end{align}
		where $\boldsymbol{1}_n$ is the vector with all elements equal to 1. Given a unit vector \( \boldsymbol{x} \in \mathbb{R}^n \) with \( \boldsymbol{x}^{T}\boldsymbol{x}=1 \), the deviation degree of \( \boldsymbol{x} \) from the null space \( N \) is defined as the square of the projection of \( \boldsymbol{x} \) onto the uniform vector \( \boldsymbol{1}_n \):
		\begin{align}
			T(\boldsymbol{x}) = (\boldsymbol{1}_n^T \boldsymbol{x})^2
		\end{align}
	\end{definition}
	
	The definition of \( T(\boldsymbol{x}) \) indicates how close a unit vector is to the null space of the uniform vector. Clearly, when \( \boldsymbol{x} \) lies within the null space \( N \), we have \( T(\boldsymbol{x}) = 0 \). As \( \boldsymbol{x} \) deviates from the null space \( N \), the value of \( T(\boldsymbol{x}) \) increases gradually. The maximum value of \( T(\boldsymbol{x}) \) occurs when \( \boldsymbol{x} \) is the uniform vector, i.e., \( \boldsymbol{x} = \pm\boldsymbol{1}_n/\sqrt{n} \), at which point \( T(\boldsymbol{x}) = n \). For the sake of description, in the following discussion, the ``deviation degree'' of a vector \( \boldsymbol{x} \) is an abbreviation for its deviation degree from the null space of the uniform vector, which refers to the value of \( T(\boldsymbol{x}) \) in Definition \ref{Tq}.
	
	\begin{proposition}
		\label{eqn}
		Denote the spectral decomposition of the sample correlation matrix by \( R^{S} = Q^{S} \Lambda^{S} (Q^{S})^{T} \), where \( \Lambda^{S} = {\rm Diag}(\lambda_{1}^{S}, \ldots, \lambda_{n}^{S}) \) is a diagonal matrix containing the \( n \) sample eigenvalues, and \( Q^{S} = (\boldsymbol{q}_{1}^{S},\ldots,\boldsymbol{q}_{n}^{S}) \) is the matrix whose columns, \( \boldsymbol{q}_{i}^{S} \), are the eigenvectors corresponding to the eigenvalue \( \lambda_{i}^{S} \). We have:
		\begin{align}
			\sum_{i=1}^{n}\lambda_{i}^{S} = \sum_{i=1}^{n}T(\boldsymbol{q}_{i}^{S}) = n
		\end{align}
	\end{proposition}
	
	\begin{proof}
		It is clear that the diagonal elements of the correlation matrix \( R^{S} \) are all equal to one. Therefore, we have:
		\begin{align}
			\lambda_{1}^{S} + \cdots + \lambda_{n}^{S} = \text{Tr}(R^{S}) = n
		\end{align}
		Since \( Q^{S}(Q^{S})^{T} = I_n \), we have:
		\begin{align}
			\sum_{i=1}^{n} T(\boldsymbol{q}_{i}^{S}) = \sum_{i=1}^{n} (\boldsymbol{1}_n^T \boldsymbol{q}_{i}^{S})^2 = \boldsymbol{1}_n^T Q^{S}(Q^{S})^{T} \boldsymbol{1}_n = \boldsymbol{1}_n^T \boldsymbol{1}_n = n
		\end{align}
		The proof is complete.
	\end{proof}
	
	
	\begin{proposition}
		\label{Tmax}
		Let \( \lambda_1^{S}, \dots, \lambda_n^{S} \) be the $n$ eigenvalues of the correlation matrix \( R^{S} \), with corresponding eigenvectors \( \boldsymbol{q}_1^{S}, \dots, \boldsymbol{q}_n^{S} \). Without loss of generality, we assume that the eigenvalues are ordered such that \( \lambda_1^{S} \leq \lambda_2^{S} \leq \dots \leq \lambda_n^{S} \). Under Assumption \ref{ass_pos}, we have:
		\begin{align}
			T(\boldsymbol{q}_{n}^{S}) \geq \lambda_n^{S} \geq n M(R^{S})
		\end{align}
		where \( M(R^{S}) = \boldsymbol{1}_n^T R^{S} \boldsymbol{1}_n/{n^2} \) represents the mean value of all elements in the correlation matrix \( R^{S} \). 
	\end{proposition}
	
	\begin{proof}
		Since $\lambda_n^{S}$ is the largest eigenvalue of the correlation matrix $R^S$, it follows that for any unit vector $\boldsymbol{x}$ (i.e., $\boldsymbol{x}^T \boldsymbol{x} = 1$), the following inequality holds:
		\begin{align}
			\lambda_n^{S} \geq \boldsymbol{x}^T R^S \boldsymbol{x}
		\end{align}
		Applying this to the eigenvector $\boldsymbol{q}_n^S$ corresponding to $\lambda_n^{S}$, we obtain:
		\begin{align}
			\lambda_n^{S} = (\boldsymbol{q}_n^S)^T R^S \boldsymbol{q}_n^S \geq \frac{\boldsymbol{1}_n^T}{\sqrt{n}} R^S \frac{\boldsymbol{1}_n}{\sqrt{n}} = \frac{\boldsymbol{1}_n^T R^S \boldsymbol{1}_n}{n} = n M(R^{S})
		\end{align}
		
		It follows from Assumption \ref{ass_pos} that $R^S$ is an irreducible positive matrix. According to the Perron-Frobenius theorem, the largest eigenvalue $\lambda_n^{S}$ (called the Perron root) is real, positive, and simple, and the corresponding eigenvector $\boldsymbol{q}_n^S$ has strictly positive components, i.e., $\boldsymbol{q}_n^S>0$. Let $E \in \mathbb{R}^{n \times n}$ be a matrix with all entries equal to one. Since $\boldsymbol{q}_n^S>0$ and the elements of the correlation matrix are between 0 and 1, we have:
		\begin{align}
			T(\boldsymbol{q}_n^S) = (\boldsymbol{q}_n^S)^T E \boldsymbol{q}_n^S \geq (\boldsymbol{q}_n^S)^T R^S \boldsymbol{q}_n^S = \lambda_n^{S}
		\end{align}
		Therefore, we conclude that:
		\begin{align}
			T(\boldsymbol{q}_n^S) \geq \lambda_n^{S} \geq n M(R^{S})
		\end{align}
		This completes the proof.
	\end{proof}
	
	
	From Proposition \ref{eqn}, we know that the sum of the \( n \) eigenvalues and the sum of the deviation degrees of the \( n \) eigenvectors both equal \( n \). According to Proposition \ref{Tmax}, for positively correlated assets, the largest eigenvalue \( \lambda_n \) is always greater than or equal to \( n M(R^S) \). This observation indicates that the largest eigenvalue is typically much larger than the others, supporting the conclusion that a dominant single strong factor exists, as discussed in \cite{Dai2024}. Furthermore, we find that the deviation degree \( T(\boldsymbol{q}_n^S) \) is greater than or equal to \( \lambda_n \), with equality holding only when all elements of \( R^S \) are equal to 1. Clearly, the condition for \( T(\boldsymbol{q}_n^S) = \lambda_n \) is stringent, suggesting a substantial gap exists between the two values. This analysis highlights that, in portfolios of positively correlated assets, the deviation degree of the largest eigenvector significantly dominates over that of the other eigenvectors. This imbalance is more pronounced than the excessive dispersion of the eigenvalues. To further illustrate this, we present a numerical comparison of \( \lambda_n^S \) and \( T(\boldsymbol{q}_n^S) \) using the ten factor-sorted portfolio datasets tested in Section \ref{PosCor}. The dominance of \( \lambda_n^S \) and \( T(\boldsymbol{q}_n^S) \) in positively correlated assets is clearly demonstrated in Figure \ref{smax}.
	\begin{figure}[!htb]
		\centering
		\includegraphics[width=1\linewidth]{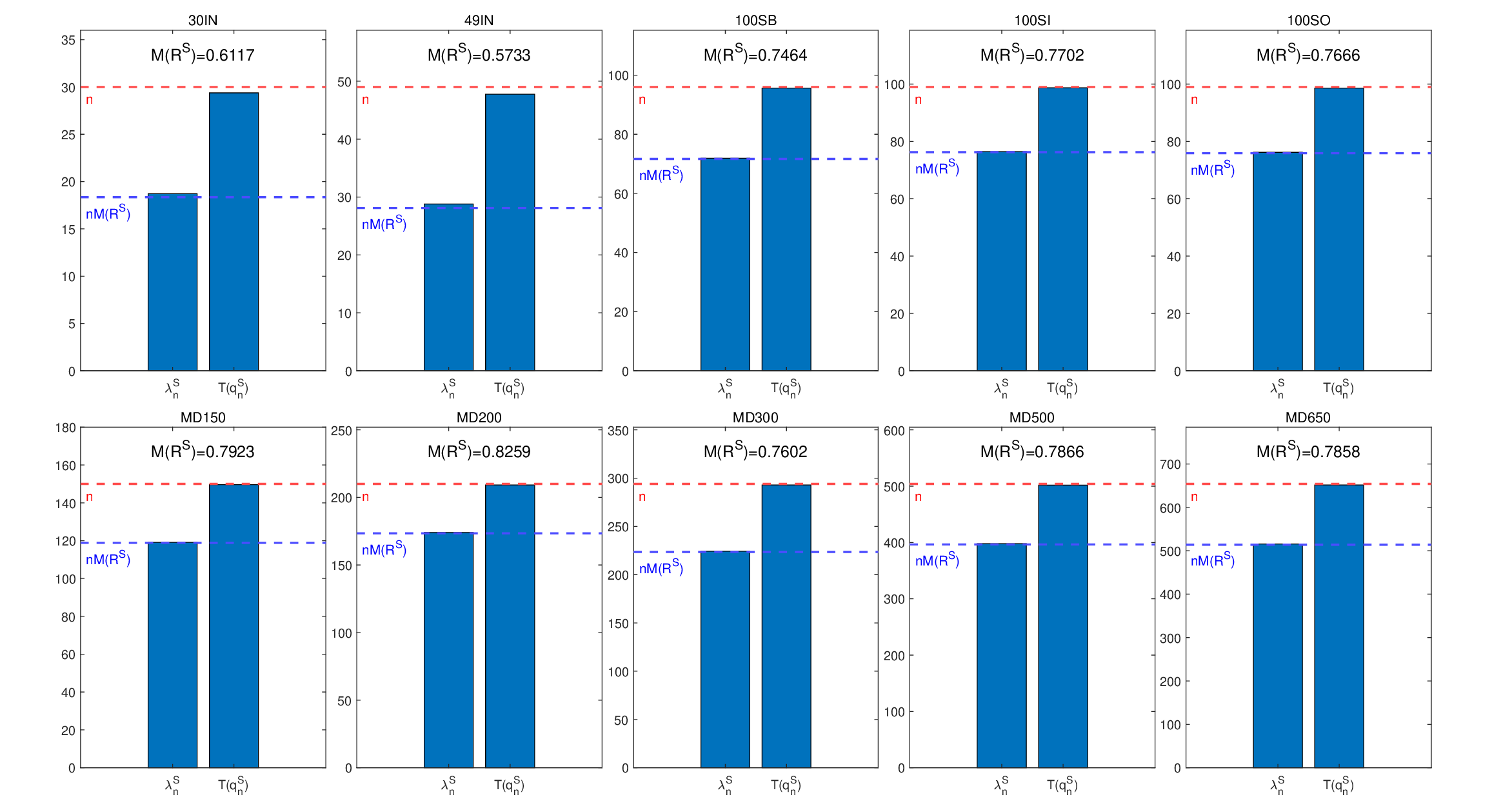}
		\caption{Values of \(\lambda_{n}^{S}\) and \(T(\boldsymbol{q}_{n}^S)\) in the Ten Factor-Sorted Portfolio Datasets}
		\label{smax}
	\end{figure}
	
	As shown in Figure \ref{smax}, for assets with high positive correlations, the deviation degree of the corresponding eigenvector, \(T(\boldsymbol{q}_{n}^S)\), is significantly greater than \(nM(R^{S})\). In particular, for high-dimensional assets, the value of \(T(\boldsymbol{q}_{n}^S)\) approaches \(n\). Note that Proposition \ref{eqn} proves that the sum of the deviation degrees of all \(n\) eigenvectors equals \(n\), implying that the deviation degrees of the remaining \(n-1\) eigenvectors are compressed to minimal values. Furthermore, we provide evidence that the eigenvectors corresponding to the weak factors are extremely close to the null space of the uniform vector.
	
	\begin{proposition}
		\label{smin}
		Let \( R^S = (\boldsymbol{\rho}_1, \dots, \boldsymbol{\rho}_n) \) represent the sample correlation coefficient matrix, where \( \boldsymbol{\rho}_i = (\rho_{i1}, \dots, \rho_{in})^{T} \) for \( i = 1, 2, \dots, n \). We assume that \( R^S \) is a positive semi-definite matrix with some zero eigenvalues. Let \( \boldsymbol{q} \) be a unit vector corresponding to an eigenvector of \( R^S \) associated with a zero eigenvalue. The deviation degree of the eigenvector \( \boldsymbol{q} \) is bounded by:
		\begin{align}
			\label{eq1}
			T(\boldsymbol{q}) \le n - \max_{1 \le i \le n} \frac{\left( \boldsymbol{1}^{T}\boldsymbol{\rho}_{i}\right)^2}{\boldsymbol{\rho}_{i}^{T}\boldsymbol{\rho}_{i}}
		\end{align}
		Furthermore, let $b =\max_{i}\min_{j}\rho_{ij}$, then we have
		\begin{align}
			\label{eq2}
			T(\boldsymbol{q}) \le n-\frac{(1+(n-1)b)^2}{1+(n-1)b^2}
		\end{align}
	\end{proposition}
	
	\begin{proof}
		Since \( \boldsymbol{q} \) is an eigenvector corresponding to a zero eigenvalue of \( R^S \), we have the following condition:
		\begin{align}
			\boldsymbol{\rho}_{i}^{T} \boldsymbol{q} = 0, \quad i = 1, \dots, n
		\end{align}
		Next, we formulate the following optimization problem:
		\begin{align}
			\max_{\boldsymbol{x}} \quad & \boldsymbol{1}_n^T \boldsymbol{x} \\
			\text{subject to} \quad & \boldsymbol{\rho}_{i}^T \boldsymbol{x} = 0 \\
			& \boldsymbol{x}^T \boldsymbol{x} = 1
		\end{align}
		The Lagrangian function for this problem is:
		\begin{align}
			\mathcal{L}(\boldsymbol{x}, \lambda, \nu) = \boldsymbol{1}_n^T \boldsymbol{x} + \lambda (\boldsymbol{\rho}_i^T \boldsymbol{x}) + \nu (\boldsymbol{x}^T \boldsymbol{x} - 1)
		\end{align}
		Setting \( \Delta \mathcal{L} = 0 \), we obtain the following system of equations:
		\begin{align}
			& \boldsymbol{1}_n + \lambda \boldsymbol{\rho}_i + 2 \nu \boldsymbol{x} = 0 \label{l1} \\
			& \boldsymbol{\rho}_i^T \boldsymbol{x} = 0 \label{l2} \\
			& \boldsymbol{x}^T \boldsymbol{x} = 1 \label{l3}
		\end{align}
		From Formula \eqref{l1}, we solve for \( \boldsymbol{x} \):
		\begin{align}
			\boldsymbol{x} = \frac{\boldsymbol{1}_n - \lambda \boldsymbol{\rho}_i}{2 \nu}
		\end{align}
		Substituting this expression into Formulas \eqref{l2} and \eqref{l3}, we obtain:
		\begin{align}
			\lambda = \frac{\boldsymbol{1}_n^T \boldsymbol{\rho}_i}{\boldsymbol{\rho}_i^T \boldsymbol{\rho}_i},\quad \nu = \pm \frac{1}{4} \sqrt{n - \frac{(\boldsymbol{1}_n^T \boldsymbol{\rho}_i)^2}{\boldsymbol{\rho}_i^T \boldsymbol{\rho}_i}}
		\end{align}
		Clearly, there are two stationary points where the objective function \( \boldsymbol{1}_n^T \boldsymbol{x} \) attains its maximum and minimum values, respectively. The values of the objective function at these points are opposites of each other. Denote \( D = \{\boldsymbol{x} \in \mathbb{R}^n : \boldsymbol{\rho}_i^T \boldsymbol{x} = 0, \boldsymbol{x}^T \boldsymbol{x} = 1\} \), then we have:
		\begin{align}
			\max_{\boldsymbol{x} \in D} \left( \boldsymbol{1}_n^T \boldsymbol{x} \right)^2 = \frac{\left( \boldsymbol{1}_n^T \boldsymbol{1}_n - \lambda \boldsymbol{1}_n^T \boldsymbol{\rho}_i \right)^2}{4 \nu^2} = n - \frac{(\boldsymbol{1}_n^T \boldsymbol{\rho}_i)^2}{\boldsymbol{\rho}_i^T \boldsymbol{\rho}_i}
		\end{align}
		Since \( \boldsymbol{q} \in D \), we have:
		\begin{align}
			\label{eq_rho}
			T(\boldsymbol{q}) = \left( \boldsymbol{1}_n^T \boldsymbol{q} \right)^2 \le \min_{1 \le i \le n} \left\{ n - \frac{(\boldsymbol{1}_n^T \boldsymbol{\rho}_i)^2}{\boldsymbol{\rho}_i^T \boldsymbol{\rho}_i} \right\} = n - \max_{1 \le i \le n} \left\{ \frac{(\boldsymbol{1}_n^T \boldsymbol{\rho}_i)^2}{\boldsymbol{\rho}_i^T \boldsymbol{\rho}_i} \right\}
		\end{align}
		This completes the proof of \eqref{eq1}.
	\end{proof}
	
	
	According to Proposition \ref{smin}, it is straightforward to derive that for an eigenvector \( \boldsymbol{q} \) corresponding to a zero eigenvalue, we have:
	\begin{align}
		\lim_{n\to \infty} T(\boldsymbol{q}) \le \lim_{n\to \infty} \frac{(n-1)(1-b)^2}{1+(n-1)b^2} = \frac{(1-b)^2}{b^2}
	\end{align}
	It is important to note that the equality condition in the above inequality is quite stringent. When the value of \( b \) is large, as is commonly observed in the ten factor-sorted portfolio datasets tested in this paper, the eigenvectors corresponding to weaker factors in high-dimensional portfolios tend to converge significantly to the null space of the uniform vector. 
	
	In summary, we observe that in high-dimensional portfolios with positive correlations, the deviation degrees of the eigenvectors of the sample correlation matrix exhibit a highly unbalanced distribution. Specifically, the deviation degree of the eigenvector associated with the largest eigenvalue dominates, while the eigenvectors corresponding to weak factors display minimal deviation degrees. This imbalance is even more pronounced than the well-documented imbalance in the sample eigenvalues. Previous research has shown that directly shrinking the differences between eigenvalues can lead to more accurate covariance matrix estimates \citep{Ledoit2012, Shi2020, Barroso2022}. Thus, a natural approach is to mitigate the eigenvectors' imbalance in their deviation degrees, which would help construct a more robust covariance matrix estimator.
	
	\section{Eigenvector Rotation Shrinkage Estimator}
	\label{model}
	
	In this section, we introduce a novel covariance matrix estimator, the Eigenvector Rotation Shrinkage Estimator (ERSE), for positively correlated assets. The core idea of ERSE is to impose a constraint on the eigenvectors of the correlation matrix within the rotation-equivariant estimation framework, ensuring that their deviation degree meets a pre-specified threshold. To achieve this, we propose the Paired Eigenvalue Rotation (PER) technique, which adjusts the deviation degrees of the eigenvectors through pairwise rotations. Specifically, the PER technique iteratively rotates the two eigenvectors with the largest difference in deviation degree. This process continues until the deviation degrees of all eigenvectors, particularly those corresponding to weak factors, reach the specified threshold, ensuring that the constraint is satisfied for the entire set of eigenvectors.
	
	We denote the \(n\) estimated eigenvectors by \( \hat{\boldsymbol{q}}_1, \hat{\boldsymbol{q}}_2, \dots, \hat{\boldsymbol{q}}_n \). The constraint on the deviation degrees of the eigenvectors is formulated as
	\begin{align}
		\label{del}
		T(\hat{\boldsymbol{q}}_{i}) \ge \delta, \quad i = 1,\ldots,n
	\end{align}
	where $\delta$ is a pre-specified threshold that defines the minimum allowable deviation degree for each eigenvector. The choice of \( \delta \) directly influences the shrinkage process by controlling the extent of the eigenvector rotations. Specifically, all eigenvectors with deviation degrees below \( \delta \) are selected as the targets for the PER technique. As discussed in Section \ref{de}, the eigenvectors corresponding to weak signals typically exhibit small deviation degrees, making them especially likely to be selected for adjustment.
	
	The threshold \( \delta \) is constrained within the interval \([0, 1]\). When \( \delta \leq 0 \), the constraint in Formula \eqref{del} becomes inactive, as \( T(\hat{\boldsymbol{q}}_i) \) is always non-negative for any eigenvector. In this case, no adjustments are made to the eigenvectors. As \( \delta \) increases, both the likelihood and the angle of rotation increase. When \( \delta = 1 \), the constraint forces all eigenvectors to satisfy \( T(\hat{\boldsymbol{q}}_1) = \cdots = T(\hat{\boldsymbol{q}}_n) = 1 \), since the condition \( \sum_{i=1}^{n} T(\hat{\boldsymbol{q}}_i) = n \) always holds (see Proposition \ref{eqn}). Finally, when \( \delta > 1 \), the constraint cannot be satisfied, as it would require the sum of the deviation degrees to exceed \( n \), which contradicts the result in Proposition \ref{eqn}.
	
	We now outline the details of implementing the Paired Eigenvalue Rotation (PER) technique. In this method, we begin by selecting the two eigenvectors with the most significant difference in their deviation degrees as the targets for rotation. The goal of the PER technique is to rotate these two eigenvectors to positions that satisfy the constraint in Formula \eqref{del} while preserving the orthogonality of the eigenvector matrix. Let \( \boldsymbol{q}_1 \) and \( \boldsymbol{q}_2 \) denote the two eigenvectors selected for rotation, where \( T(\boldsymbol{q}_1) < \delta < T(\boldsymbol{q}_2) \) and \( T(\boldsymbol{q}_1) + T(\boldsymbol{q}_2) > 2\delta \)\footnote{These two conditions ensure that eigenvector \( \boldsymbol{q}_1 \) must be adjusted, with at least one feasible solution available to satisfy the constraint in Formula \eqref{del} for both eigenvectors.}. We define the rotation process of the paired eigenvectors as follows:
	\begin{align}
		&\hat{\boldsymbol{q}}_1 = \cos \theta\cdot\boldsymbol{q}_1  + \sin \theta\cdot\boldsymbol{q}_2 \label{ra1}\\
		&\hat{\boldsymbol{q}}_2 = -\sin \theta\cdot\boldsymbol{q}_1 + \cos \theta\cdot\boldsymbol{q}_2 \label{ra2}
	\end{align}
	where \( \theta \) represents the rotation angle. Figure \ref{ro} visually illustrates the rotation process described above. It can be seen that orthogonality between arbitrary pairs of eigenvectors is preserved because the two eigenvectors, \( \boldsymbol{q}_1 \) and \( \boldsymbol{q}_2 \), are rotated pairwise within the plane they lie in. In other words, by replacing the original eigenvectors with the rotated ones, all the eigenvectors still form an orthogonal basis. Another key observation is that this rotation process preserves the total deviation of the two eigenvectors from the uniform vector. Specifically, the rotation increases the deviation degree of \( \boldsymbol{q}_1 \) while decreasing that of \( \boldsymbol{q}_2 \). This property of the PER technique is formally captured in the following proposition:
	\begin{figure}[!htb]
		\centering
		\subfloat[2-Dimensional]{\includegraphics[width=2.9in]{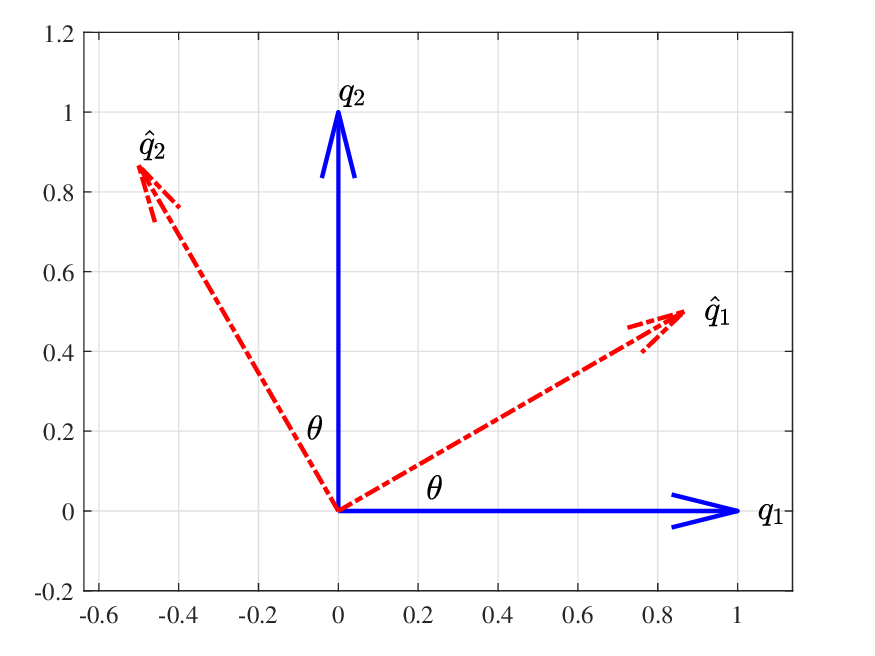}
			\label{ro1}}
		\hfil
		\subfloat[3-Dimensional]{\includegraphics[width=2.9in]{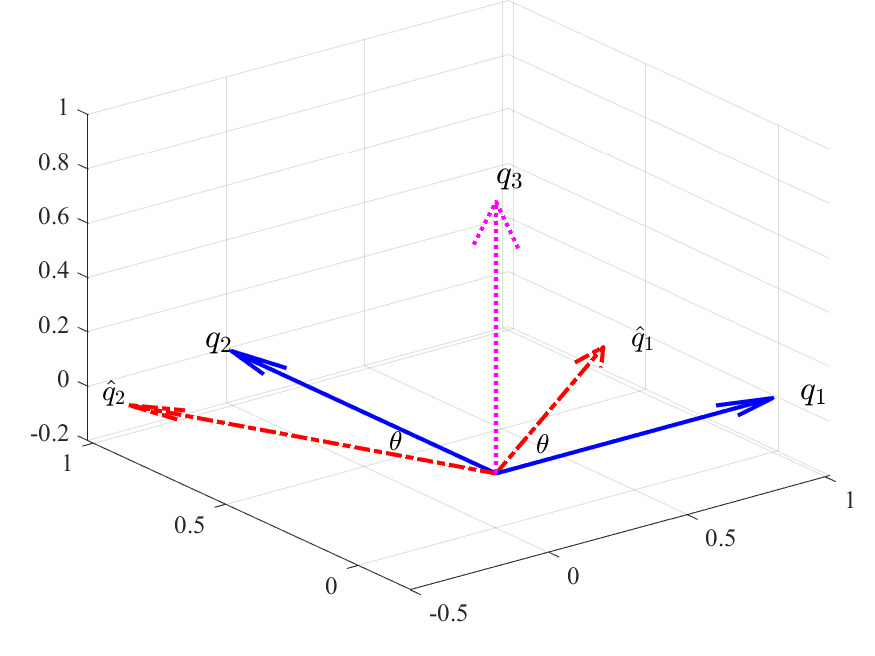}
			\label{ro2}}
		
		\caption{Visualization of the Paired Eigenvector Rotation}
		\label{ro}
	\end{figure}
	\begin{proposition}
		\label{sumT}
		Under the rotation process represented by Formulas \eqref{ra1} and \eqref{ra2}, we have:
		\begin{enumerate}[(1)]
			\item $\hat{\boldsymbol{q}}_1^{T}\hat{\boldsymbol{q}}_1 = \hat{\boldsymbol{q}}_2^{T}\hat{\boldsymbol{q}}_2 = 1$ and $\hat{\boldsymbol{q}}_1^{T}\hat{\boldsymbol{q}}_2=0$ \label{po1}.
			\item For any eigenvector \( \boldsymbol{q} \) among the remaining \( n-2 \) eigenvectors not selected for rotation, 
			\( \boldsymbol{q}^T \hat{\boldsymbol{q}}_1 = \boldsymbol{q}^T \hat{\boldsymbol{q}}_2 = 0 \)\label{po2}.
			\item $T(\hat{\boldsymbol{q}}_1) + T(\hat{\boldsymbol{q}}_2) = T(\boldsymbol{q}_1) + T(\boldsymbol{q}_2)$. \label{po3}
		\end{enumerate}
	\end{proposition}
	\begin{proof}
		\label{apro1}
		The statements \ref{po1} and \ref{po2} in the Proposition \ref{sumT} can be easily proven by substituting Formulas \eqref{ra1} and \eqref{ra2} into the expressions for \( \hat{\boldsymbol{q}}_1^{T}\hat{\boldsymbol{q}}_1 \), \( \hat{\boldsymbol{q}}_2^{T}\hat{\boldsymbol{q}}_2 \), \( \hat{\boldsymbol{q}}_1^{T}\hat{\boldsymbol{q}}_2 \), \( \boldsymbol{q}^{T}\hat{\boldsymbol{q}}_1 \), and \( \boldsymbol{q}^{T}\hat{\boldsymbol{q}}_2 \). We now detail the proof of statement \ref{po3}:
		According to Formula \eqref{ra1}, we can easily derive:
		\begin{align}
			T(\hat{\boldsymbol{q}}_1) &= \left( \cos\theta \cdot \boldsymbol{1}_n^T \boldsymbol{q}_1 + \sin\theta \cdot \boldsymbol{1}_n^T \boldsymbol{q}_2 \right)^2 \nonumber\\
			&= \cos^2\theta \cdot (\boldsymbol{1}_n^T \boldsymbol{q}_1)^2 + \sin^2\theta \cdot (\boldsymbol{1}_n^T \boldsymbol{q}_2)^2 + 2\cos\theta \sin\theta \cdot \boldsymbol{1}_n^T \boldsymbol{q}_1 \cdot \boldsymbol{1}_n^T \boldsymbol{q}_2
		\end{align}
		Similarly, using Formula \eqref{ra2}, we obtain:
		\begin{align}
			T(\hat{\boldsymbol{q}}_2) = \sin^2\theta \cdot (\boldsymbol{1}_n^T \boldsymbol{q}_1)^2 + \cos^2\theta \cdot (\boldsymbol{1}_n^T \boldsymbol{q}_2)^2 - 2\cos\theta \sin\theta \cdot \boldsymbol{1}_n^T \boldsymbol{q}_1 \cdot \boldsymbol{1}_n^T \boldsymbol{q}_2
		\end{align}
		Therefore, we have:
		\begin{align}
			T(\hat{\boldsymbol{q}}_1) + T(\hat{\boldsymbol{q}}_2) &= (\cos^2\theta + \sin^2\theta) T(\boldsymbol{q}_1) + (\cos^2\theta + \sin^2\theta) T(\boldsymbol{q}_2)= T(\boldsymbol{q}_1) + T(\boldsymbol{q}_2)
		\end{align}
		The proof is complete.
	\end{proof}
	
	
	Another important consideration in the PER technique is determining the appropriate rotation angle. The primary goal of the rotation is to ensure that the eigenvector \( \hat{\boldsymbol{q}}_1 \) satisfies the constraint in Formula \eqref{del}. To achieve this, we choose the smallest rotation angle that satisfies the constraint. We let: 
	\begin{align}
		T(\hat{\boldsymbol{q}}_1) = \left( \boldsymbol{1}_n^T \boldsymbol{q}_1 \cos \theta + \boldsymbol{1}_n^T \boldsymbol{q}_2 \sin \theta \right)^2 = \delta
	\end{align}
	Consequently, using inverse trigonometric functions, we can easily obtain two feasible solutions for the rotation angle $\theta\in\left[-\pi/{2}, \pi/{2}\right]$:
	\begin{align}
		\label{theta}
		\theta = \arctan\left(\frac{-s_{1}s_{2}\pm \sqrt{\delta(s_{1}^2 + s_{2}^2-\delta)}}{s_{2}^{2}-\delta}\right)
	\end{align}
	where $s_{1} = \boldsymbol{1}_n^T \boldsymbol{q}_1$ and $s_{2} = \boldsymbol{1}_n^T \boldsymbol{q}_2$. Note that two possible angles in Formula \eqref{theta} satisfy the rotation requirement, corresponding to the \(\pm\) choices\footnote{Since the two selected eigenvectors satisfy the condition $T(\boldsymbol{q}_1)<\delta<T(\boldsymbol{q}_2)$, i.e., $s_{1}^{2}<\delta<s_{2}^{2}$, the existence of these two feasible solutions is guaranteed.}. Since we always prefer a smaller rotation, we select the angle with the smaller absolute value between the two angles.
	
	To summarize, we present the main steps of the PER technique in Algorithm \ref{PER}. It can be seen that the PER technique allows the application of a rotation angle to an eigenvector (i.e., \( \boldsymbol{q}_1 \)) associated with a weak factor that does not satisfy the constraint in Formula \eqref{del}, adjusting its deviation degree to the preset threshold.  Simultaneously, to maintain the orthogonality of the eigenvector matrix, we rotate the accompanying eigenvector (i.e., \( \boldsymbol{q}_2 \)) to reduce its deviation degree. In practice, we select the eigenvector with the largest deviation degree as the accompanying eigenvector.
	
	\begin{algorithm}[!htb]\small
		\caption{PER($\boldsymbol{q}_{1},\boldsymbol{q}_{2},\delta$)} 
		\label{PER}
		\begin{algorithmic}[1]
			\Require $\delta$: a threshold for the minimum deviation. $\boldsymbol{q}_{1}, \boldsymbol{q}_{2}$: two eigenvectors such that \( T(\boldsymbol{q}_1) < \delta < T(\boldsymbol{q}_2) \) and \( T(\boldsymbol{q}_1) + T(\boldsymbol{q}_2) > 2\delta \).
			\Ensure $\hat{\boldsymbol{q}}_1, \hat{\boldsymbol{q}}_2$: two new eigenvectors after rotation.
			\State Let $T(\boldsymbol{q}_1) = \delta$, and solve for the two possible rotation angles using the inverse function:
			\[
			\theta_1 = \arctan\left(\frac{-s_{1}s_{2}+ \sqrt{\delta(s_{1}^2 + s_{2}^2-\delta)}}{s_{2}^{2}-\delta}\right)\quad {\rm and}\quad \theta_2 = \arctan\left(\frac{-s_{1}s_{2}-\sqrt{\delta(s_{1}^2 + s_{2}^2-\delta)}}{s_{2}^{2}-\delta}\right),
			\]
			where $s_{1} = \boldsymbol{1}_n^T \boldsymbol{q}_1$ and $s_{2} = \boldsymbol{1}_n^T \boldsymbol{q}_2$.
			\State Select an rotation angle as the one with the smallest absolute value between the above two angles:
			\begin{align*}
				\theta = \left\{\begin{aligned}
					&\theta_1,\quad \textrm{if}\quad  |\theta_1| \le |\theta_2|;\\
					&\theta_2,\quad \textrm{if}\quad |\theta_1| > |\theta_2|.\\
				\end{aligned}
				\right.
			\end{align*}
			\State Rotate $\boldsymbol{q}_1$ and $\boldsymbol{q}_2$ with respect to angle $\theta$:
			\[
			\hat{\boldsymbol{q}}_1 = \cos \theta\cdot\boldsymbol{q}_1  + \sin \theta\cdot\boldsymbol{q}_2
			\quad {\rm and} \quad
			\hat{\boldsymbol{q}}_2 = -\sin \theta\cdot\boldsymbol{q}_1 + \cos \theta\cdot\boldsymbol{q}_2.
			\]
			\State \textbf{Return:} $\hat{\boldsymbol{q}}_1$ and $\hat{\boldsymbol{q}}_2$.
		\end{algorithmic}
	\end{algorithm}
	
	Next, we introduce the implementation of the ERSE method, which constructs the covariance matrix estimator by incorporating the PER technique as its core mechanism. The key steps of the ERSE algorithm are summarized in Algorithm \ref{ERSE}. First, a spectral decomposition of the sample correlation matrix is performed to obtain the sample eigenvectors. The initial estimators for the \( n \) eigenvectors are then set to the sample eigenvectors. Next, the following iterative procedure is carried out until all the \( n \) eigenvectors satisfy the constraint in Formula \eqref{del}: identify the eigenvector with the smallest deviation degree, pair it with the eigenvector with the largest deviation degree, and apply the PER operator to rotate these two eigenvectors. Finally, using the reconstructed eigenvector estimators, the covariance matrix estimator is constructed within the rotation-equivariant estimation framework.
	\begin{algorithm}[!htb]\small
		\caption{ERSE($R^{S},D^{S},\delta$)}
		\label{ERSE}
		\begin{algorithmic}[1]
			\Require $R^{S}$: the sample correlation matrix; $D^{S}$: the diagonal matrix with the sample standard deviations on the diagonal; $\delta$: a threshold for the minimum deviation.
			\Ensure $\hat{\Sigma}$: the estimation of covariance matrix. 
			\State Perform the spectral decomposition on the correlation matrix $R^{S}$ to obtain the $n$ sample eigenvectors $Q^{S}=(\boldsymbol{q}_{1}^{S},\ldots,\boldsymbol{q}_{n}^{S})$.
			\State Initialize the estimators of the $n$ sample eigenvectors and the subscript of the eigenvector to be rotated:
			\[
			\hat{\boldsymbol{q}}_{i} \leftarrow \boldsymbol{q}_{i}^{S}, \quad i = 1,\ldots,n \quad {\rm and}\quad i_{\min} \leftarrow \arg\min_{1\le i \le n}T(\hat{\boldsymbol{q}}_{i}).
			\]
			\While{$T(\boldsymbol{q}_{i_{\min}})<\delta$}
			\State Select the accompanying eignvector with the largest deviation: 
			\[
			i_{\max} = \arg\max_{1\le i \le n}T(\hat{\boldsymbol{q}}_{i}).
			\]
			\State Apply the PER operator to the eigenvectors \( \boldsymbol{q}_{i_{\min}} \) and \( \boldsymbol{q}_{i_{\max}} \):
			\[
			\hat{\boldsymbol{q}}_{i_{\min}},\hat{\boldsymbol{q}}_{i_{\max}}\leftarrow {\rm PER}(\hat{\boldsymbol{q}}_{i_{\min}},\hat{\boldsymbol{q}}_{i_{\max}},\delta).
			\]
			\State Update the subscript of the eigenvector to be rotated:
			\[
			i_{\min} \leftarrow \arg\min_{1\le i \le n}T(\hat{\boldsymbol{q}}_{i})
			\]
			\EndWhile
			\State Reconstruct the estimated eigenvalues \( \hat{\boldsymbol{\lambda}}=(\hat{\lambda}_{1},\ldots,n)^{T} \) using the rotated eigenvectors:
			\[
			\hat{\lambda}_{i} = \hat{\boldsymbol{q}}_{i}^{T} R^{S} \hat{\boldsymbol{q}}_{i},\quad i=1,\ldots,n.
			\]
			\State Reconstruct the estimated covariance matrix:
			\[
			\hat{\Sigma} = D^{S} Q^{S}{\rm Diag}(\hat{\boldsymbol{\lambda}}) (Q^{S})^{T} D^{S}.
			\]
			\State \textbf{Return:} $\hat{\Sigma}$.
		\end{algorithmic}
	\end{algorithm}
	
	Finally, we discuss the computational cost of the proposed ERSE. The computational complexity of the PER technique is cheap since it mainly depends on a simple inverse trigonometric process. In the ERSE algorithm, the PER technique is applied a maximum of \(n-1\) times. This limitation arises because the eigenvector associated with the largest eigenvalue consistently exceeds the specified threshold. In extreme cases, the PER technique may be applied once to each of the remaining \(n-1\) eigenvectors. Therefore, the proposed ERSE method exhibits excellent computational efficiency and is well-suited for high-dimensional assets.
	
	\section{Relation to Existing Methods}
	\label{rela}
	As previously mentioned, the proposed ERSE method can be regarded as a modified approach to covariance matrix shrinkage estimation. It shares strong connections with rotation-equivariant estimation techniques discussed in the existing literature, such as \cite{Ledoit2004, Ledoit2012, Ledoit2017,Ledoit2020, Shi2020, Nguyen2022}. 
	
	Overall, the proposed ERSE method distinguishes itself from existing approaches in two main ways. First, the ERSE method is specifically designed for covariance matrix estimation of positively correlated assets, while previous methods do not explicitly target this scenario. This focus on ERSE provides an advantage in this specific context, although it has limitations in broader applicability. Second, unlike existing methods that directly shrink the sample covariance matrix, the ERSE method focuses on shrinking the sample correlation matrix. The correlation matrix is derived from the covariance matrix by standardizing the covariances, offering the advantage of providing a unit-free measure. It is important to clarify that we do not claim that shrinking the correlation matrix is inherently superior to shrinking the covariance matrix. Instead, we emphasize that focusing on the correlation matrix provides a more convenient analytical path for positively correlated assets.
	
	Next, we detail the technical differences and similarities between the ERSE method and classical rotation-equivariant estimation approaches. It is important to note that most existing shrinkage methods focus on the covariance matrix, while our approach targets the correlation matrix. In the following discussion, we use $\Sigma^S$ to denote the sample covariance matrix when referring to other methods and the sample correlation matrix in the context of our proposed method. Similarly, $\Sigma$ will represent the true covariance matrix or the true correlation matrix, depending on the context. We will not explain the distinction between the two matrices but will focus on analyzing their eigenvalues and eigenvectors. 
	
	Within the rotation-equivariant estimation framework, \cite{Ledoit2022} provides an optimal eigenvalue estimator that relies on the true covariance matrix:
	\begin{align}
		\label{opt}
		\hat{\lambda}_{i}^{*} = (\boldsymbol{q}_{i}^{S})^{T} \Sigma \boldsymbol{q}_{i}^{S},\quad i=1,\ldots,n
	\end{align}
	The estimator outlined in Formula \eqref{opt} relies on the unobservable matrix $\Sigma$, which is not available in practice. However, it offers valuable theoretical insights and serves as a useful reference for constructing practical estimators. Starting from the sample covariance matrix, the eigenvalues of this textbook estimator can be rewritten as ${\lambda}_{i}^{S} = (\boldsymbol{q}_{i}^{S})^{T} \Sigma^S \boldsymbol{q}_{i}^{S},\quad i=1,\ldots,n.$
	
	A class of linear shrinkage estimation methods proposes replacing the sample covariance matrix $\Sigma^S$ with a linear combination of $\Sigma^S$ and a target matrix. For instance, \cite{Ledoit2004} suggests using the identity matrix as the target matrix, while \cite{Ledoit2003} advocates for a target matrix derived from the capital asset pricing model. In contrast, the proposed ERSE method retains the sample covariance matrix $\Sigma^S$ and instead enhances the robustness of the estimator by modifying the eigenvectors: $\hat{\lambda}_{i} = (\hat{\boldsymbol{q}}_{i})^{T} \Sigma^S \hat{\boldsymbol{q}}_{i},\quad i=1,\ldots,n$, where the estimated eigenvectors, i.e., $\hat{\boldsymbol{q}}_{1},\ldots,\hat{\boldsymbol{q}}_{n}$, are obtained using Algorithm \ref{ERSE}. It is clear that if \( \delta \le \min_{1 \le i \le n} T(\boldsymbol{q}_{i}^{S}) \), the output of the proposed ERSE method will revert to the sample covariance matrix as no modifications are made to the sample eigenvectors in this case.
	
	The proposed ERSE method also exhibits a linear shrinkage pattern similar to that presented in \cite{Ledoit2004}, with the specific details outlined in the following proposition.
	\begin{proposition}
		\label{line}
		Based on the rotation process defined by Formulas \eqref{ra1} and \eqref{ra2}, the two eigenvalues exhibit the following linear relationship:
		\begin{align}
			\hat{\lambda}_{1} = \gamma \lambda_{1} + (1-\gamma) \lambda_{2} \quad {\rm and}\quad \hat{\lambda}_{2} = (1-\gamma) \lambda_{1} + \gamma \lambda_{2}, \quad {\rm where} \quad \gamma = \cos^{2} \theta \in [0, 1]
		\end{align}
	\end{proposition}
	
	\begin{proof}
		According to Formulas \eqref{ra1} and \eqref{ra2}, we have
		\begin{align}
			\nonumber
			\hat{\lambda}_1 &= \hat{\boldsymbol{q}}_1^{T} R^{S} \hat{\boldsymbol{q}}_1 = ( \cos\theta\cdot\boldsymbol{q}_1 + \sin \theta\cdot\boldsymbol{q}_2)^{T}R^{S}(\cos \theta\cdot\boldsymbol{q}_1  + \sin \theta\cdot\boldsymbol{q}_2 ) \\
			& = \cos^{2}\theta\cdot\boldsymbol{q}_1^{T}R^{S}\boldsymbol{q}_1 + \sin^{2}\theta\cdot\boldsymbol{q}_2^{T}R^{S}\boldsymbol{q}_2 = \cos^{2}\theta\cdot\lambda_{1} + \sin^{2}\theta\cdot\lambda_{2}
		\end{align}
		Similarly, we have: $\hat{\lambda}_2 = \sin^{2}\theta\cdot\lambda_{1} + \cos^{2}\theta\cdot\lambda_{2}$. The proof is complete.
	\end{proof}

	
	\cite{Ledoit2004} proposes a shrinkage function of the form \( \hat{\lambda}_{i}^{LW} = \gamma_{LW} \lambda_{i} + (1-\gamma_{LW}) \mu \), where \( \mu \) is the mean of all sample eigenvalues. Both the proposed ERSE method and the method in \cite{Ledoit2004} operate within a linear framework that reduces the disparity between eigenvalues while preserving the sum of all eigenvalues. However, there are two critical differences between the two approaches: First, in \cite{Ledoit2004}, the shrinkage direction for all eigenvalues is uniformly set to the mean \( \mu \), whereas the ERSE method selects two eigenvalues with the largest deviations as targets for mutual shrinkage. Second, while \cite{Ledoit2004} applies a uniform shrinkage coefficient \( \gamma_{LW} \) to all eigenvalues, the ERSE method uses different shrinkage coefficients for each eigenvalue, with the values determined by the deviation of the corresponding eigenvectors. From the perspective of assigning different shrinkage coefficients to different eigenvalues, the ERSE method is more similar to the nonlinear shrinkage techniques, such as those in \cite{Ledoit2012} and \cite{Ledoit2017}.
	
	In addition, from the portfolio perspective, the proposed ERSE method is related to the $\ell_2$-norm constraint discussed in \cite{DeMiguel2009}. \cite{DeMiguel2009} introduces the concept of incorporating an \( \ell_2 \)-norm constraint into the GMV portfolio model to enhance the stability of the portfolio structure. From an economic perspective, the eigenvectors can be viewed as the portfolio vectors of $n$ virtual funds, with their corresponding eigenvalues representing the risk associated with the returns of these funds. Building on this understanding, we present the following proposition to illustrate the relationship between the ERSE method and the \( \ell_2 \)-norm constraint.
	\begin{proposition}
		\label{norm2}
		Given an eigenvector \( \boldsymbol{q}_i \), we define its corresponding unit-cost portfolio \( \boldsymbol{w}(\boldsymbol{q}_i) = \boldsymbol{q}_i / ( \boldsymbol{1}_n^T \boldsymbol{q}_i) \), where the denominator \( \boldsymbol{1}_n^T \boldsymbol{q}_i \) is assumed to be nonzero. The minimum deviation constraint \( T(\boldsymbol{q}_i) \geq \delta \) is equivalent to imposing an $\ell_2$-norm constraint on the unit-cost portfolio \( \boldsymbol{w}(\boldsymbol{q}_i) \), i.e.,
		\begin{align}
			T(\boldsymbol{q}_i) \geq \delta \quad \Leftrightarrow \quad \| \boldsymbol{w}(\boldsymbol{q}_i) \|_2 \leq \frac{1}{\sqrt{\delta}}.
		\end{align}
	\end{proposition}
	\begin{proof}
		Based on the Definition \ref{Tq}, we have 
		\begin{align}
			\|\boldsymbol{w}(\boldsymbol{q}_i) \|_2^{2} = \left(\frac{\boldsymbol{q}_i}{\boldsymbol{1}_n^T \boldsymbol{q}_i}\right)^{T}\cdot \frac{\boldsymbol{q}_i}{\boldsymbol{1}_n^T \boldsymbol{q}_i} = \frac{\boldsymbol{q}_i^{T}\boldsymbol{q}_i}{\boldsymbol{1}_n^T \boldsymbol{q}_i \boldsymbol{q}_i^{T}\boldsymbol{1}_n} = \frac{1}{T(\boldsymbol{q}_i)}
		\end{align}
		Since $\|\boldsymbol{w}(\boldsymbol{q}_i) \|_2\ge 0$ and $T(\boldsymbol{q}_i)\ge 0$, we have
		\begin{align}
			T(\boldsymbol{q}_i) \geq \delta \quad \Leftrightarrow \quad \| \boldsymbol{w}(\boldsymbol{q}_i) \|_2^{2} \leq \frac{1}{\delta} \quad \Leftrightarrow \quad \| \boldsymbol{w}(\boldsymbol{q}_i) \|_2 \leq \frac{1}{\sqrt{\delta}}
		\end{align}
		The proof is complete.
	\end{proof}
	
	
	According to the above proposition, we observe that imposing a deviation constraint on the eigenvector is equivalent to imposing an \( \ell_2 \)-norm constraint on its corresponding unit-cost portfolio. It follows that the deviation degree \( T(\boldsymbol{q}_{i}) \) defined in Formula \eqref{Tq} is inversely proportional to the \( \ell_2 \)-norm of the unit-cost portfolio corresponding to \( \boldsymbol{q}_i \). 
	
	\section{Empirical Results}
	\label{res}
	To evaluate the out-of-sample performance of the proposed ERSE method, we compare it with state-of-the-art rotation-equivariant estimators from the existing literature. We use the ten datasets of factor-sorted portfolios discussed in Section \ref{PosCor}, all of which exhibit persistently positive correlation matrices. In the comparison step, we focus on estimating the Global Minimum Variance (GMV) portfolio without imposing short-selling constraints. This approach provides an effective means to evaluate the quality of the covariance matrix estimator, a common method in the experimental part of the covariance matrix estimation literature \citep{Ledoit2017, Shi2020, Cui2024}. The optimal portfolio of GMV model is formulated as:
	\begin{equation}
		\label{wgmv}
		\boldsymbol{w}^{*}(\Sigma) = \frac{\Sigma^{-1}\boldsymbol{1}_{n}}{\boldsymbol{1}_{n}^{T}\Sigma^{-1} \boldsymbol{1}_{n}}
	\end{equation} 
	In practice, the natural approach is replacing the unknown value of true covariance matrix \( \Sigma \) with an estimator \( \hat{\Sigma} \) in equation \eqref{wgmv} to obtain a feasible portfolio.

	\subsection{Alternatives Methods}
	In our analysis, we consider ten alternative strategies for the ERSE method, outlined as follows:
	
	\begin{enumerate}
		\item \textbf{EW:} the equally weighted portfolio promoted by \cite{DeMiguel2009a}.  
		\item \textbf{Sample:} the GMV portfolio where $\hat{\Sigma}_{t}$ is given by the sample covariance matrix. 
		\item \textbf{LIN1P:} the linear shrinkage towards a two-parameter matrix; all variances are identical, and all covariances are zero (See \cite{Ledoit2004}).
		\item \textbf{LIN2P:} the linear shrinkage towards a two-parameter matrix; all variances are identical, and all covariances are the same (See Appendix B.1 of \cite{Ledoit1995}).
		\item \textbf{LINC:} the linear shrinkage towards a constant-correlation matrix; the target preserves the diagonal of the sample covariance matrix, and all correlation coefficients are identical (See \cite{Ledoit2004a}).
		\item \textbf{LIND:} the linear shrinkage towards a diagonal matrix; the target preserves the diagonal of the sample covariance matrix, with all covariances set to zero (See Appendix B.2 of \cite{Ledoit1995}).
		\item \textbf{LINM:} the linear shrinkage towards a one-factor market model, where the factor is defined as the cross-sectional average of all random variables. Due to the idiosyncratic volatility of the residuals, the target preserves the diagonal of the sample covariance matrix (See \cite{Ledoit2003}).
		\item \textbf{LIS:} the nonlinear shrinkage derived under Stein’s loss, known as linear-inverse shrinkage (See Section 3 of \cite{Ledoit2022a}).
		\item \textbf{QIS:} the nonlinear shrinkage derived under Frobenius loss and its two variants—Inverse Stein’s loss and Minimum Variance loss—referred to as quadratic-inverse shrinkage (See Section 4.5 of \cite{Ledoit2022a}).
		\item \textbf{GIS:} The nonlinear shrinkage derived under the Symmetrized Kullback-Leibler loss. This method can be viewed as geometrically averaging linear-inverse shrinkage (LIS) with quadratic-inverse shrinkage (QIS) (See Remark 4.3 of \cite{Ledoit2022a}).
	\end{enumerate}
	
	The above algorithms covers a range of linear and nonlinear shrinkage techniques. These methods are implemented using the publicly available ``covShrinkage'' project provided by Ledoit's official repository\footnote{For more details, please refer to: https://github.com/oledoit/covShrinkage}. Default values are used for all parameters to ensure consistency across the evaluated methods. A key parameter in the ERSE estimation method is the minimum threshold for the deviation degree, denoted \(\delta\). In this study, we set \(\delta = 0.25\), which is empirically determined and has proven effective in most cases. To demonstrate that the choice of \(\delta\) is not entirely post hoc, we also examine the strategy's performance under different \(\delta\) values in subsequent analyses. Further details can be found in Subsection \ref{ddr}.
	
	We use the rolling window approach to evaluate portfolio performance. First, an estimation window is defined to perform the necessary computations. Following precedents in the literature, such as \cite{DeMiguel2009} and \cite{Shi2020}, this study adopts an estimation window of $L=120$ data points, corresponding to ten years of monthly data. At each time step, different covariance matrix estimators and portfolio weights are computed using the historical return data from the most recent $L$ months. The rolling window approach is then repeated by incorporating data from the subsequent month while excluding data from the earliest month. This iterative process continues until the end of the data set is reached. Consequently, $T-L$ portfolio returns are generated for each portfolio strategy, where $T$ is the total number of periods in the dataset.
	
	\subsection{Out-of-Sample Variance}
	\label{res_var}
	The most crucial performance metric for the GMV portfolio is the out-of-sample variance, which directly quantifies the portfolio's risk. The GMV model aims to minimize portfolio variance (or standard deviation) without considering expected returns. Therefore, when evaluating a GMV strategy, the primary focus should be its effectiveness in minimizing risk. In this context, out-of-sample variance serves as a quantitative measure for assessing the covariance matrix estimator, with lower variance indicating a better covariance matrix estimator \citep{Ledoit2017, Shi2020}. We also assess whether one portfolio has a statistically significantly lower out-of-sample variance than another. Our analysis is limited to comparisons between the ERSE method and other alternative approaches to mitigate the risk of multiple testing issues. For each scenario, the $p$ value for the null hypothesis of equal variance is computed using the stationary bootstrap method proposed by \cite{Ledoit2011}, with bootstrap samples set to $B=1,000$ and a block size of $b=5$. Table \ref{var} presents the experimental results on out-of-sample risk.
	\begin{table}[!htb]\small
		\centering
		\caption{Out-of-Sample Variance}\label{var}
		\resizebox{\textwidth}{!}{
			\begin{tabular}{lllllllllll}
				\toprule
				& 30IN & 49IN & 100SB & 100SI & 100SO & 150MD & 200MD & 300MD & 500MD & 650MD \\
				\midrule
				EW  & 22.55$^{***}$ & 23.02$^{***}$ & 27.15$^{***}$ & 26.40$^{***}$ & 26.68$^{***}$ & 23.48$^{***}$ & 26.72$^{***}$ & 26.70$^{***}$ & 26.68$^{***}$ & 25.86$^{***}$ \\
				SAM & 13.50$^{**}$ & 14.02$^{***}$ & 18.77$^{***}$ & 17.88$^{***}$ & 18.77$^{***}$ & NA & NA & NA & NA & NA \\
				LIN1P & 12.99$^{*}$ & 12.93$^{***}$ & 15.77$^{***}$ & 14.88$^{***}$ & 14.86$^{***}$ & 13.54$^{*}$ & 8.77 & 17.01$^{***}$ & 10.11$^{***}$ & 10.37$^{***}$ \\
				LIN2P & 12.86$^{*}$ & 12.66$^{***}$ & 16.00$^{***}$ & 15.12$^{***}$ & 15.02$^{***}$ & 13.78$^{**}$ & 9.98$^{*}$ & 17.34$^{***}$ & 10.92$^{***}$ & 11.05$^{***}$ \\
				LINC & 12.80 & 12.71$^{**}$ & 14.58$^{*}$ & 14.30$^{***}$ & 14.19$^{**}$ & 12.73 & 7.91 & 13.26$^{**}$ & 8.01$^{***}$ & 8.14$^{***}$ \\
				LIND & 13.04$^{*}$ & 13.10$^{***}$ & 15.84$^{***}$ & 14.82$^{***}$ & 15.02$^{***}$ & 13.65$^{**}$ & 8.81 & 17.14$^{***}$ & 10.00$^{***}$ & 10.29$^{***}$ \\
				LINM & 12.94 & 12.55$^{**}$ & 15.94$^{***}$ & 14.78$^{***}$ & 14.79$^{***}$ & 14.07$^{**}$ & 9.26$^{*}$ & 15.78$^{***}$ & 9.09$^{***}$ & 9.51$^{***}$ \\
				GIS & 13.03$^{**}$ & 12.65$^{***}$ & 14.34 & 13.64$^{**}$ & 13.34 & NA & NA & NA & NA & NA \\
				LIS & 13.02$^{**}$ & 12.66$^{***}$ & 14.23 & 13.54$^{**}$ & 13.31 & NA & NA & NA & NA & NA \\
				QIS & 13.03$^{**}$ & 12.65$^{***}$ & 14.47 & 13.74$^{**}$ & 13.42 & 14.31$^{**}$ & 8.94$^{**}$ & 12.57$^{*}$ & 8.28$^{***}$ & 9.05$^{***}$ \\
				ERSE & \textbf{12.28} & \textbf{11.93} & \textbf{13.38} & \textbf{12.25} & \textbf{12.57} & \textbf{11.96} & \textbf{6.96} & \textbf{11.05} & \textbf{6.98} & \textbf{7.04} \\
				\bottomrule
			\end{tabular}
		}
	\end{table}
	
	The results of the out-of-sample variance comparisons among different strategies highlight the superior performance of the ERSE method in risk management. Specifically, the GMV strategy based on the ERSE method achieves the lowest variance across all datasets, with statistically significant results in most cases. This demonstrates the effectiveness of the ERSE method in improving covariance matrix estimation accuracy and reducing portfolio risk. All shrinkage estimation strategies outperform the traditional equally weighted (EW) and sample average (SAM) strategies, with this advantage becoming more pronounced as the number of assets increases. These results support the well-established effectiveness of shrinkage estimation techniques in covariance matrix estimation for high-dimensional assets. Among the linear shrinkage methods, the LINC method performs best. The LINC method leverages the ideas of the CAPM to capture consistent trends across assets, giving it an advantage over other linear shrinkage estimators in the context of positively correlated assets. The ERSE method achieves an average risk reduction of 10.52\% relative to LINC across all 10 datasets, with statistically significant improvements observed in 7 of the 10 datasets. Compared to the QIS method, representative of nonlinear shrinkage techniques, the ERSE method achieves an average risk reduction of 12.46\% across all datasets. The above empirical results highlight that the ERSE estimator outperforms other shrinkage techniques in reducing out-of-sample portfolio risk, demonstrating its advantage in decision scenarios characterized by positively correlated assets.
	
	It is important to note that the out-of-sample performance reported here depends on the selected dataset. We do not claim that the ERSE method will consistently outperform all benchmark strategies in every setting. However, the experimental results provide strong evidence that the ERSE method significantly improves covariance matrix estimation for factor-sorted portfolios, demonstrating its effectiveness in decision-making environments characterized by positively correlated assets.
	
	\subsection{Different Thresholds for Minimum Deviation Degree}
	\label{ddr}
	
	The experimental results discussed in the previous section were obtained with \( \delta = 0.25 \). Since \( \delta \) is a critical hyperparameter for the ERSE method, evaluating its out-of-sample performance over different values of \( \delta \) is important. \( \delta \) is allowed to vary between 0 and 1. In particular, when \( \delta = 0 \), the ERSE strategy transforms into the SAM strategy, which is unsuitable for high-dimensional asset portfolios. Therefore, we set a lower bound for \( \delta \) at 0.05 and study its influence on the strategy's performance within the interval \( [0.05, 1] \). We also present the out-of-sample risk for other comparative strategies in Figure \ref{difS}, although their performance remains unaffected by changes in \( \delta \).
	\begin{figure}[!htb]
		\centering
		\includegraphics[width=\linewidth]{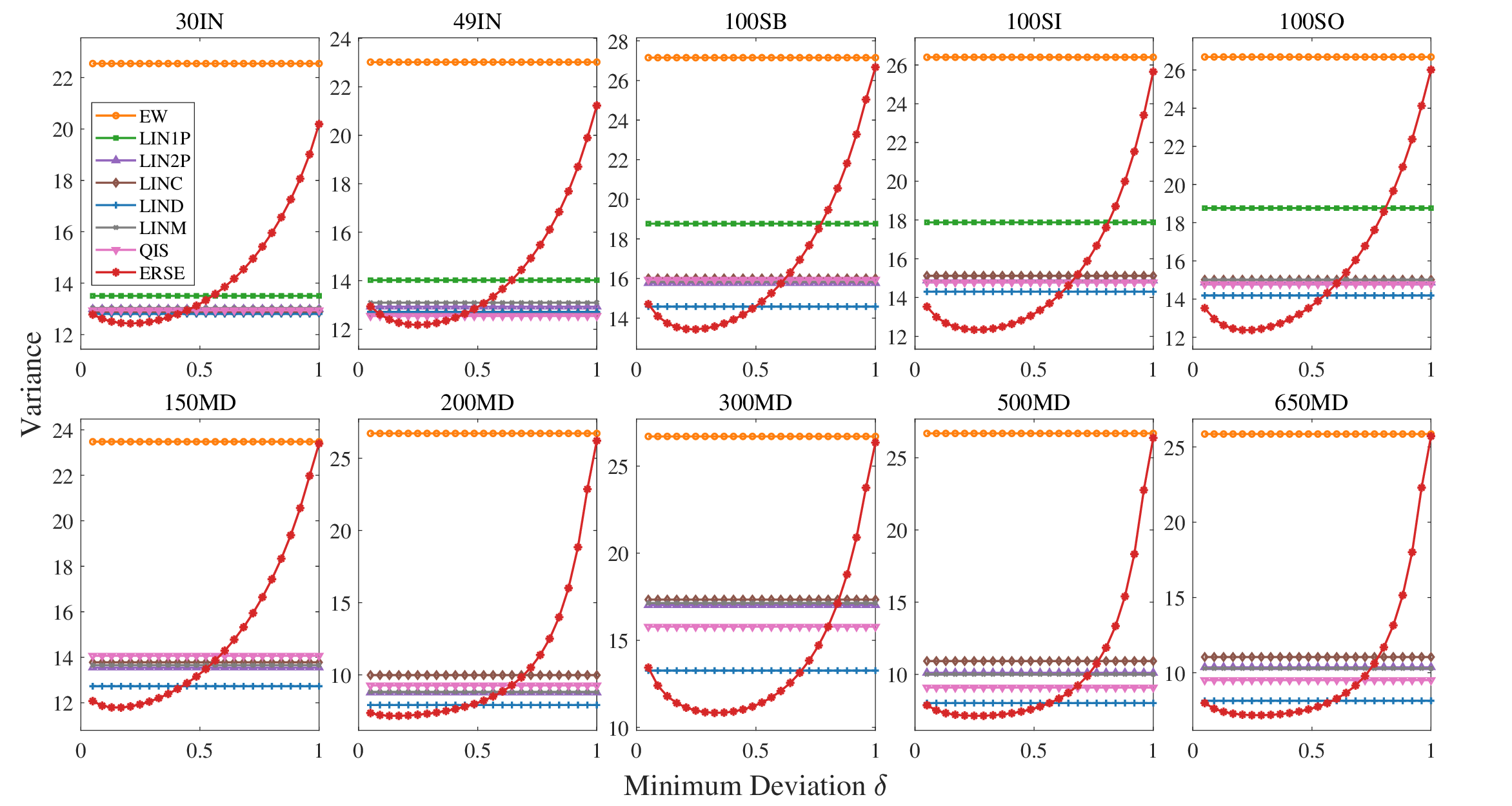}
		\caption{Out-of-Sample Variance of Strategies across Varying Threshold of $\delta$}
		\label{difS}
	\end{figure}
	
	As shown in Figure \ref{difS}, the out-of-sample risk of the ERSE strategy exhibits a non-monotonic relationship with \( \delta \), initially decreasing and then increasing as \( \delta \) rises. When \( \delta \) is at its minimum value of 0.05, the ERSE strategy performs well on all ten datasets, demonstrating the substantial improvement in out-of-sample performance achieved by a modest rotation of the sample eigenvectors. As \( \delta \) increases, the ERSE strategy’s out-of-sample performance of the ERSE strategy improves progressively, peaking when \( \delta \) falls between 0.15 and 0.35. Further increases in \( \delta \) lead to a sharp rise in out-of-sample risk, ultimately resulting in a performance that lags slightly behind the EW strategy. In particular, the ERSE strategy outperforms the competing strategies when \( \delta \) is set within a moderate range, and this range expands as the dimensionality of the asset space increases. A practical and reliable range for the parameter setting has been identified as \( \delta \in [0.15, 0.35] \), which we recommend as the optimal parameter interval. In conclusion, the ERSE strategy shows an advantage in determining the optimal shrinkage coefficient. By choosing \( \delta\) within this appropriate range, superior out-of-sample performance can be consistently achieved across a wide range of datasets for positively correlated assets.
	
	\subsection{Stability of the Estimated Covariance Matrix and Portfolios}
	We evaluate the stability of the proposed ERSE strategy by examining the condition numbers of the covariance matrix estimators and the portfolio weights derived from the GMV models. The condition number of the covariance matrix is a crucial indicator of the numerical stability associated with its inverse \citep{Won2013}. Mathematically, it is measured by the ratio of the largest eigenvalue to the smallest eigenvalue. A high condition number indicates an ill-conditioned matrix close to singularity, making the inversion process highly sensitive to significant numerical errors. We present the average condition numbers for various covariance matrix estimators in Table \ref{cod}. The results show that the SAM strategy has high condition numbers that increase significantly as the number of assets in the portfolio rises, indicating poor numerical stability. In contrast, the ERSE strategy consistently yields the lowest condition numbers across all datasets, while minimizing both the mean and standard deviation of the condition number. These results highlight the superior robustness and numerical stability of the covariance matrix estimates produced by the ERSE method, reinforcing its potential as a reliable tool for risk management and portfolio optimization.
	
	\begin{table}[!htb]\small
		\centering
		\caption{Average Characteristics of Condition Numbers}\label{cod}
		\resizebox{\textwidth}{!}{
			\begin{tabular}{lllllllllll}
				\toprule
				& 30IN & 49IN & 100SB & 100SI & 100SO & MD150 & MD200 & MD300 & MD500 & MD650 \\
				\midrule
				\multicolumn{11}{c}{\textbf{Mean}}\\
				SAM   & 408.01 & 976.46 & 7.17E3 & 7.52E3 & 7.40E3 & NA & NA & NA & NA & NA \\
				LIN1P & 273.44 & 518.54 & 2.24E3 & 2.16E3 & 2.21E3 & 8.44E3 & 1.36E4 & 1.40E4 & 2.96E4 & 3.84E4 \\
				LIN2P & 259.73 & 451.09 & 2.61E3 & 2.57E3 & 2.55E3 & 9.88E3 & 3.35E4 & 1.61E4 & 4.62E4 & 5.72E4 \\
				LINC  & 224.59 & 444.98 & 1.15E3 & 1.15E3 & 1.16E3 & 4.12E3 & 1.19E4 & 5.44E3 & 2.14E4 & 2.67E4 \\
				LIND  & 288.01 & 595.59 & 2.34E3 & 2.22E3 & 2.30E3 & 1.21E4 & 2.66E4 & 1.80E4 & 6.73E4 & 8.63E4 \\
				LINM  & 315.61 & 577.22 & 2.44E3 & 2.19E3 & 2.31E3 & 1.95E4 & 4.65E4 & 1.34E4 & 8.39E4 & 1.33E5 \\
				GIS   & 277.47 & 526.24 & 1.81E3 & 1.39E3 & 1.40E3 & NA & NA & NA & NA & NA \\
				LIS   & 277.99 & 531.77 & 1.90E3 & 1.43E3 & 1.52E3 & NA & NA & NA & NA & NA \\
				QIS   & 276.97 & 520.88 & 1.75E3 & 1.35E3 & 1.38E3 & 9.90E4 & 6.16E5 & 1.22E4 & 1.88E5 & 1.45E8 \\
				ERSE  & \textbf{111.30} & \textbf{247.01} & \textbf{388.50} & \textbf{415.89} & \textbf{389.28} & \textbf{845.69} & \textbf{1614.08} & \textbf{1624.55} & \textbf{4428.55} & \textbf{5650.04} \\
				\midrule
				\multicolumn{11}{c}{\textbf{Standard Deviation}}\\
				SAM   & 182.56 & 609.29 & 8.62E3 & 1.21E4 & 1.30E4 & NA & NA & NA & NA & NA  \\
				LIN1P & 65.42 & 140.01 & 212.58 & 363.98 & 343.48 & 2.53E3 & 4.68E3 & 2.99E3 & 9.91E3 & 1.28E4 \\
				LIN2P & 63.94 & 127.90 & 570.11 & 693.86 & 739.34 & 1.64E3 & 7.79E3 & 1.25E3 & 8.95E3 & 1.13E4 \\
				LINC  & 30.85 & 95.31 & 126.83 & 133.09 & 104.30 & 921.53 & 2.35E3 & 989.43 & 4.89E3 & 6.15E3 \\
				LIND  & 73.37 & 203.06 & 309.36 & 500.65 & 407.32 & 3.03E3 & 6.94E3 & 1.58E3 & 1.81E4 & 2.39E4 \\
				LINM  & 105.31 & 180.86 & 555.57 & 673.33 & 651.79 & 3.81E3 & 1.90E4 & 3.31E3 & 3.51E4 & 3.14E4 \\
				GIS   & 76.66 & 186.60 & 141.05 & 255.71 & 326.69 & NA & NA & NA & NA & NA  \\
				LIS   & 77.46 & 193.16 & 236.42 & 347.78 & 1.08E3 & NA & NA & NA & NA & NA  \\
				QIS   & 75.90 & 180.36 & 135.69 & 195.34 & 181.12 & 4.78E5 & 3.19E5 & 5.00E4 & 4.29E5 & 3.36E9 \\
				ERSE  & \textbf{14.77} & \textbf{53.99} & \textbf{42.05} & \textbf{62.17} & \textbf{43.72} & \textbf{34.36} & \textbf{98.41} & \textbf{181.13} & \textbf{253.80} & \textbf{308.39 }\\
				\bottomrule
		\end{tabular}}
	\end{table}
	
	An alternative approach to assessing the stability of covariance matrix estimators is to analyze the portfolio weights of the GMV strategies, as ill-conditioned covariance matrix estimators are often linked to extreme positions in individual assets, leading to undesirable portfolio allocations. Figure \ref{difw} displays the distribution of portfolio weights for the strategies across different datasets. As shown in the figure, the ERSE strategy consistently produces more concentrated and stable weight distributions, with the absolute value of almost all portfolio weights remaining below 0.5. This indicates that the ERSE approach generates less volatile and extreme portfolio weights. The ERSE strategy maintains relatively narrow weight distributions in datasets with more assets, highlighting its robustness and reduced dependence on sample size. In contrast, strategies like SAM and GIS exhibit broader, more extreme portfolio weights, particularly as the number of assets increases, suggesting that these methods are more susceptible to instability in covariance matrix estimates. Additionally, the ERSE strategy leads to more balanced investment allocations across assets, reinforcing its diversification advantage. These results align with the findings of \cite{DeMiguel2009}, which uses the $\ell_2$-norm constraint to emphasize diversified portfolios. A surprising finding is that the ERSE strategy has a shorter left tail in its portfolio weight distribution compared to other strategies. This suggests that substantial short-selling of individual assets is undesirable in the context of positively correlated assets. Excessive short-selling implies high leverage, which inherently increases risk. While the ERSE strategy does incorporate a moderate level of small-scale short-selling, this is consistent with the intuition of using hedging investments in positively correlated assets to mitigate risk.
	\begin{figure}[!htb]
		\centering
		\includegraphics[width=1\linewidth]{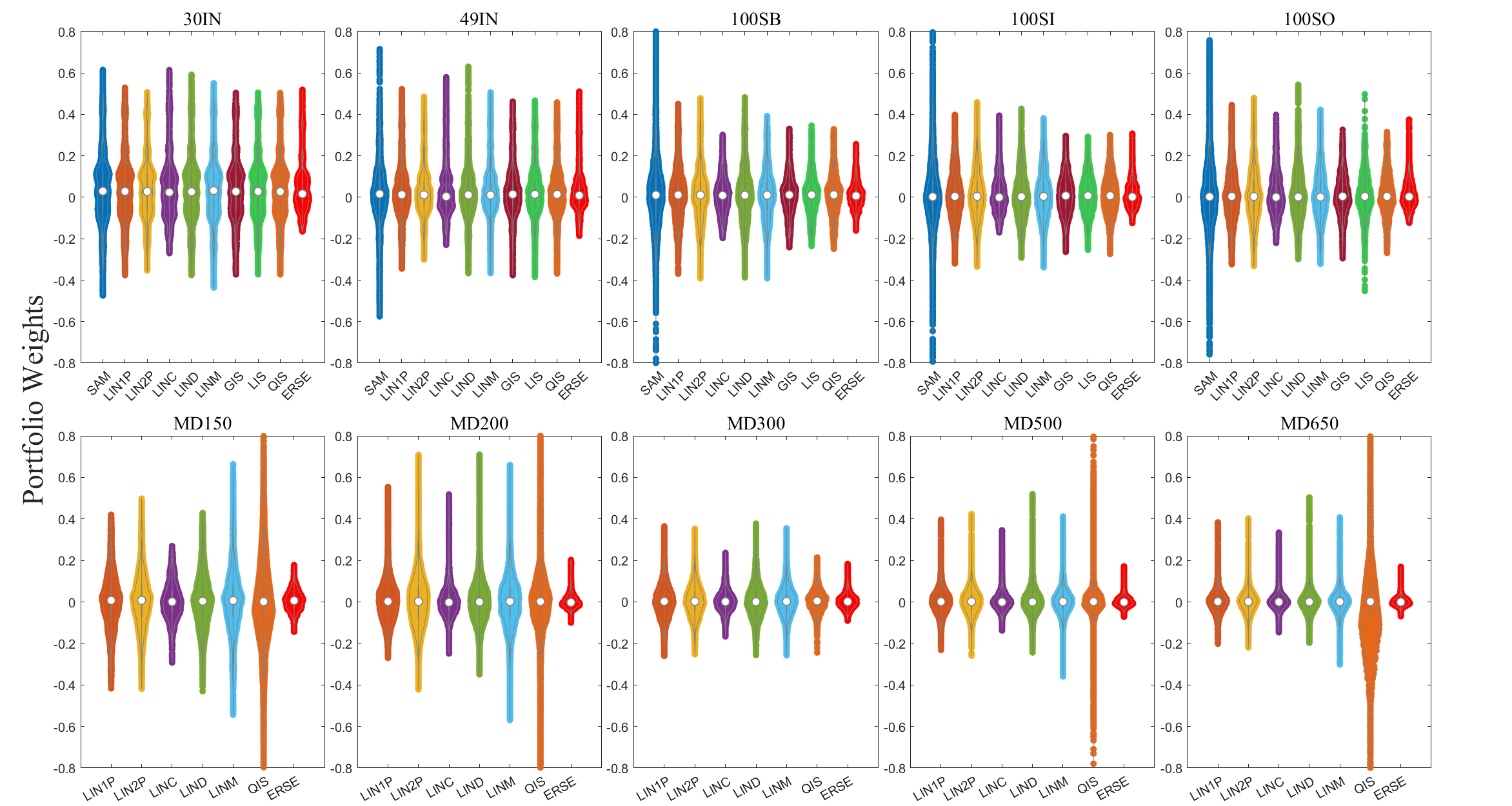}
		\caption{Distribution of Portfolio Weights}
		\label{difw}
	\end{figure}
	
	\subsection{Sharpe Ratio}
	The tradeoff between risk and return remains a central consideration in portfolio management. The Sharpe ratio, calculated as the ratio of mean excess returns to their standard deviation, provides insight into how effectively each strategy compensates for risk and is a key tool for evaluating portfolio performance. To test the statistical significance of the differences in Sharpe ratios between the ERSE strategy and alternative portfolios, we use the bootstrap method proposed by \cite{Ledoit2008}, with bootstrap samples set to $B=1,000$ and a block size of $b=5$. Table \ref{sr} reports the monthly out-of-sample Sharpe ratios of different strategies across various datasets.
	\begin{table}[!htb]\small
		\centering
		\caption{Out-of-Sample Sharpe ratio}\label{sr}
		\resizebox{\textwidth}{!}{
			\begin{tabular}{lllllllllll}
				\toprule
				& 30IN & 49IN & 100SB & 100SI & 100SO & 150MD & 200MD & 300MD & 500MD & 650MD \\
				\midrule
				EW & \textbf{0.3763} & \textbf{0.3508} & 0.2963$^{***}$ & 0.3033$^{***}$ & 0.2850$^{***}$ & 0.1494$^{***}$ & \textbf{0.1525}$^{***}$ & 0.1332$^{***}$ & \textbf{0.1527}$^{***}$ & 0.1557$^{***}$ \\
				SAM & 0.3330 & 0.2590$^{***}$ & 0.3283 & 0.3566 & 0.2968$^{*}$ & NA & NA & NA & NA & NA \\
				LIN1P & 0.3357 & 0.2793$^{***}$ & 0.3333 & 0.3511 & 0.3073 & 0.1955 & 0.1113$^{***}$ & 0.1516$^{**}$ & 0.1258$^{***}$ & 0.1477$^{***}$ \\
				LIN2P & 0.3360 & 0.2856$^{***}$ & 0.3324 & 0.3527 & 0.3058 & 0.1975 & 0.1120 & 0.1509$^{**}$ & 0.1244$^{***}$ & 0.1484$^{***}$ \\
				LINC & 0.3310 & 0.2877$^{**}$ & 0.3303 & 0.3464 & \textbf{0.3141} & 0.1871 & 0.1295$^{*}$ & \textbf{0.1576} & 0.1426$^{**}$ & 0.1560$^{*}$ \\
				LIND & 0.3346 & 0.2744$^{***}$ & 0.3323 & 0.3462 & 0.3073 & 0.1949 & 0.1195$^{***}$ & 0.1508$^{**}$ & 0.1334$^{***}$ & 0.1550$^{**}$ \\
				LINM & 0.3328 & 0.2823$^{**}$ & 0.3315 & 0.3499 & 0.3023 & 0.1995 & 0.1156$^{***}$ & 0.1513$^{**}$ & 0.1259$^{***}$ & 0.1443$^{***}$ \\
				GIS & 0.3355 & 0.2843$^{**}$ & \textbf{0.3364} & 0.3517 & 0.3087 & NA & NA & NA & NA & NA \\
				LIS & 0.3353 & 0.2840$^{**}$ & 0.3352 & 0.3520 & 0.3076 & NA & NA & NA & NA & NA \\
				QIS & 0.3356 & 0.2846$^{**}$ & 0.3371 & 0.3516 & 0.3108 & \textbf{0.2027} & 0.1115$^{***}$ & 0.1558 & 0.1338$^{***}$ & 0.1512$^{***}$ \\
				ERSE & 0.3446 & 0.3021 & 0.3177 & \textbf{0.3391} & 0.2946 & 0.1729 & 0.1353 & 0.1518 & 0.1470 & \textbf{0.1573} \\
				\bottomrule
		\end{tabular}}
	\end{table}
	
	As shown in Table \ref{sr}, the EW strategy outperforms the other strategies in four of the ten datasets, consistent with the findings in \cite{DeMiguel2009a}. Additionally, the ERSE strategy performs exceptionally well in the 100SI and 650MD datasets, achieving the highest Sharpe ratios in these cases while delivering competitive performance in others. In most datasets, we observe no significant differences in the Sharpe ratios between the ERSE strategy and other competing methods. This result aligns closely with \cite{Goto2015} and \cite{Shi2020}, which note that detecting statistical differences in Sharpe ratios across alternative portfolio strategies is challenging due to substantial estimation errors in expected returns. Furthermore, the ERSE strategy achieves higher Sharpe ratios than the other strategies in the 650MD dataset, with most comparisons being statistically significant. This superior performance can be attributed to the better covariance matrix estimation provided by the ERSE method.
	
	\subsection{Subperiod Analysis}
	The out-of-sample period spans $T-L$ months, and it is important to consider whether specific subperiods drive the superior performance of the ERSE strategy. To address this concern, we divide the out-of-sample period into two subperiods and repeat the experimental analysis for each one. The results are presented in Table \ref{subp}. Table \ref{subp} shows that the ERSE strategy outperforms its competitors in terms of out-of-sample variance over both subperiods for all ten datasets. Statistically significant differences in variance are particularly evident in the high-dimensional datasets, further reinforcing the strategy's effectiveness in managing risk. The results from the subperiod analysis are consistent with those from the full-period analysis in Subsection \ref{res_var}, highlighting the robustness of the ERSE strategy across different time intervals. 
	\begin{table}[!htb]\small
		\centering
		\caption{Out-of-Sample Variances for Subperiod Analysis}\label{subp}
		\resizebox{\textwidth}{!}{
			\begin{tabular}{lllllllllll}
				\toprule
				& 30IN & 49IN & 100SB & 100SI & 100SO & 150MD & 200MD & 300MD & 500MD & 650MD \\
				\midrule
				\multicolumn{11}{c}{\textbf{Subperiod 1}} \\
				EW & 20.82$^{***}$ & 22.01$^{***}$ & 23.92$^{***}$ & 24.06$^{***}$ & 24.47$^{***}$ & 21.30$^{***}$ & 23.81$^{***}$ & 24.11$^{***}$ & 23.95$^{***}$ & 23.24$^{***}$ \\
				SAM & 13.35$^{*}$ & 14.19$^{***}$ & 21.31$^{***}$ & 17.99$^{***}$ & 21.75$^{***}$ & NA & NA & NA & NA & NA \\
				LIN1P & 12.54 & 12.34$^{**}$ & 16.00$^{***}$ & 12.84$^{***}$ & 14.86$^{**}$ & 11.79 & 7.33 & 14.34$^{***}$ & 8.76$^{**}$ & 9.02$^{**}$ \\
				LIN2P & 12.40 & 12.06$^{**}$ & 16.53$^{***}$ & 13.27$^{***}$ & 15.33$^{**}$ & 12.22 & 8.23 & 15.21$^{***}$ & 9.62$^{**}$ & 9.74$^{**}$ \\
				LINC & 12.56 & 12.09 & 14.77 & 13.05$^{*}$ & 14.40$^{**}$ & 10.93 & 7.56 & 11.35$^{**}$ & 7.39 & 7.40 \\
				LIND & 12.59 & 12.61$^{*}$ & 16.24$^{***}$ & 12.81$^{***}$ & 15.10$^{**}$ & 11.99 & 7.32 & 14.68$^{***}$ & 8.63$^{**}$ & 8.81$^{*}$ \\
				LINM & 12.51 & 11.88 & 16.69$^{***}$ & 13.24$^{***}$ & 15.17$^{**}$ & 12.85 & 8.27$^{**}$ & 14.49$^{***}$ & 8.38$^{**}$ & 8.60$^{**}$ \\
				GIS & 12.58$^{**}$ & 11.99$^{**}$ & 13.92 & 11.41$^{**}$ & 12.76 & NA & NA & NA & NA & NA \\
				LIS & 12.58$^{**}$ & 12.01$^{**}$ & 13.84 & 11.30$^{**}$ & 12.75 & NA & NA & NA & NA & NA \\
				QIS & 12.58$^{**}$ & 11.96$^{**}$ & 14.04 & 11.55$^{***}$ & 12.87 & 13.01$^{***}$ & 8.17$^{***}$ & 10.27$^{**}$ & 7.40$^{**}$ & 7.64$^{**}$ \\
				ERSE & \textbf{11.56} & \textbf{11.17} & \textbf{13.07} & \textbf{9.68} & \textbf{11.93} & \textbf{9.76} & \textbf{6.67} & \textbf{9.52} & \textbf{6.78} & \textbf{6.72} \\
				\midrule
				\multicolumn{11}{c}{\textbf{Subperiod 2}} \\
				EW & 24.29$^{***}$ & 24.05$^{***}$ & 30.37$^{***}$ & 28.74$^{***}$ & 28.91$^{***}$ & 25.65$^{***}$ & 29.66$^{***}$ & 29.30$^{***}$ & 29.42$^{***}$ & 28.48$^{***}$ \\
				SAM & 13.52 & 13.86$^{**}$ & 16.01$^{*}$ & 17.34$^{*}$ & 15.84$^{**}$ & NA & NA & NA & NA & NA \\
				LIN1P & 13.29 & 13.48$^{**}$ & 15.26 & 16.56$^{*}$ & 14.88$^{**}$ & 14.92 & 10.20 & 19.44$^{***}$ & 11.26$^{***}$ & 11.54$^{***}$ \\
				LIN2P & 13.17 & 13.21 & 15.21 & 16.61$^{*}$ & 14.73$^{**}$ & 14.92 & 11.71 & 19.23$^{***}$ & 12.03$^{***}$ & 12.18$^{***}$ \\
				LINC & 12.89 & 13.26 & 14.15 & 15.27 & 13.91 & 14.22 & 8.19 & 14.94 & 8.49$^{***}$ & 8.73$^{***}$ \\
				LIND & 13.33 & 13.55$^{**}$ & 15.16 & 16.50$^{*}$ & 14.95$^{**}$ & 14.94 & 10.28 & 19.36$^{***}$ & 11.19$^{***}$ & 11.61$^{***}$ \\
				LINM & 13.21 & 13.15 & 14.91 & 15.99$^{*}$ & 14.42 & 14.85 & 10.24 & 16.90$^{**}$ & 9.71$^{***}$ & 10.33$^{***}$ \\
				GIS & 13.32 & 13.27$^{*}$ & 14.44 & 15.54 & 13.88 & NA & NA & NA & NA & NA \\
				LIS & 13.32 & 13.25 & 14.32 & 15.46 & 13.84 & NA & NA & NA & NA & NA \\
				QIS & 13.33 & 13.28 & 14.57 & 15.62 & 13.93 & 15.09 & 9.69 & 14.62 & 9.03$^{*}$ & 10.33$^{**}$ \\
				ERSE & \textbf{12.85} & \textbf{12.61} & \textbf{13.48} & \textbf{14.59} & \textbf{13.15} & \textbf{13.96} & \textbf{7.18} & \textbf{12.36} & \textbf{7.05} & \textbf{7.24} \\
				\bottomrule
		\end{tabular}}
	\end{table}
	
	\subsection{Different Estimation Windows}
	In the above analysis, the covariance matrix estimators are constructed using the most recent \( L = 120 \) monthly returns. To assess the robustness of our results, we conduct additional tests using alternative estimation windows, allowing \( L \) to vary between 60 and 240, representing 5 to 20 years of historical data. To ensure comparability across different estimation windows, we standardize the starting investment period at the 241st data point. Figure \ref{difL} presents the results.
	\begin{figure}[htb]
		\centering
		\includegraphics[width=\textwidth]{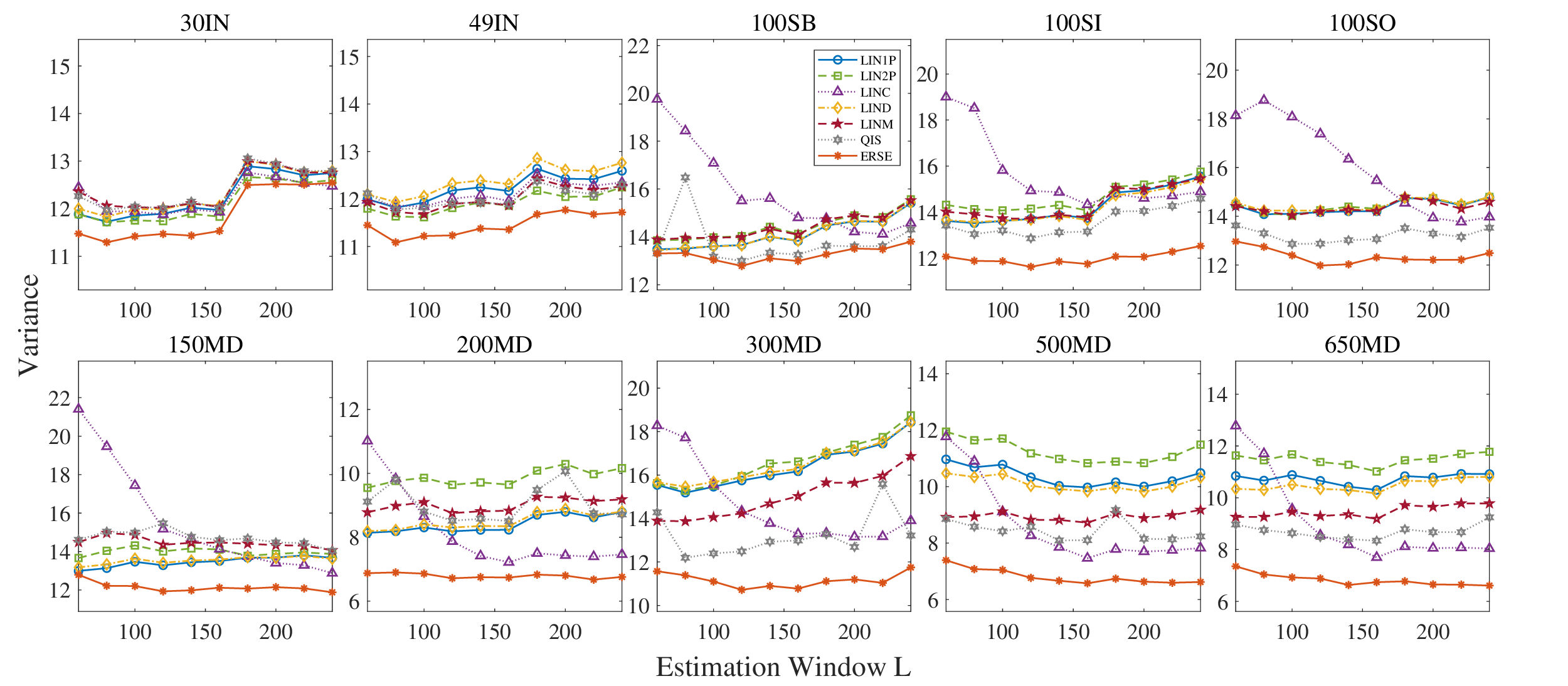}
		\caption{Portfolio risks with Different Estimation Windows of $L$}
		\label{difL}
	\end{figure}
	
	The results show that the ERSE strategy consistently outperforms the competing strategies in terms of portfolio risk for \( L \in [60, 240] \). In contrast, the performance of the other strategies exhibits greater variability as the estimation window \( L \) changes. Specifically, the LINC strategy requires a relatively larger estimation window when applied to portfolios with higher asset dimensions; otherwise, its out-of-sample performance deteriorates significantly. As discussed in Subsection \ref{res_var}, the LINC strategy performs well with positively correlated assets, often ranking just behind the ERSE strategy across multiple datasets. However, one of the key advantages of the ERSE strategy, as highlighted in Figure \ref{difL}, is its robustness to changes in the estimation window. The ERSE strategy shows less sensitivity to variations in sample size, providing stable and reliable covariance matrix estimators across different estimation windows.
	
	\subsection{Randomized Samples of Assets}
	Given that the 650MD dataset encompasses the majority of assets in our analysis, we use it as the sampling universe for robustness testing. We conduct 150 iterations of random subsampling, drawing 200 assets in each iteration to generate distinct datasets for stability and out‐of‐sample performance assessment. Comparative experiments of different strategies are then conducted on these 150 datasets. This approach allows us to evaluate the performance of the ERSE strategy across a wide range of asset combinations, providing insight into its robustness when applied to different subsets of available assets. Additionally, we calculated the win rate of the ERSE strategy over the alternative strategies, defined as the number of times the ERSE strategy outperformed its counterparts. We also calculated the average deviation between the ERSE strategy and each alternative strategy over the 150 tests and statistically tested the hypothesis that this deviation is zero. Table \ref{ran_a} presents the results, which include the mean, standard deviation, maximum, and minimum values of the out-of-sample variances across the 150 tests. As shown in Table \ref{ran_a}, the ERSE strategy has the lowest mean, standard deviation, maximum, and out-of-sample risks across 150 tests. Additionally, the ERSE strategy achieved a 100\% win rate against the alternative strategies over the 150 tests. In summary, the ERSE strategy demonstrates superior and robust performance across a wide range of randomly sampled datasets, making it a highly effective approach for portfolio optimization in the context of positively correlated assets.
	\begin{table}[!htb]\small
		\centering
		\caption{Summary Statistic of the Out-of-Sample Variance for 150 Randomized Samples}\label{ran_a}
		\resizebox{\textwidth}{!}{
			\begin{tabular}{lllllllll}
				\toprule
				& EW & LIN1P & LIN2P & LINC & LIND & LINM & QIS & ERSE \\
				\midrule
				MEAN & 25.8868 & 9.8741 & 10.4614 & 8.8552 & 9.9255 & 9.9334 & 9.8867 & \textbf{8.0620 }\\
				STD  & 0.5348 & 0.8371 & 0.9037 & 0.7464 & 0.8779 & 0.8859 & 2.5388 & \textbf{0.5224} \\
				MAX  & 27.6446 & 12.2836 & 12.9950 & 11.2524 & 12.2697 & 12.3987 & 35.1709 & \textbf{9.7579} \\
				MIN  & 24.5760 & 8.1918 & 8.6419 & 7.2846 & 8.2038 & 8.2637 & 7.6268 & \textbf{7.0075} \\
				WR   & 100\% & 100\% & 100\% & 100\% & 100\% & 100\% & 100\% & NA \\
				MD   & 17.8248$^{***}$ & 1.8121$^{***}$ & 2.3994$^{***}$ & 0.7932$^{***}$ & 1.8635$^{***}$ & 1.8714$^{***}$ & 1.8247$^{***}$ & NA \\
				\bottomrule
		\end{tabular}}
	\end{table}
	
	
	\section{Conclusion}
	\label{conclu}
	
	This paper addresses the problem of estimating the covariance matrix for positively correlated assets. Empirical evidence from Ken-French factor-sorted portfolios confirms the presence of strong positive correlations and motivates our research. The proposed ERSE estimator enforces a minimum deviation from the uniform vector and uses a PER procedure to adjust paired eigenvectors while preserving orthogonality. Theoretically, we show that the leading eigenvector departs more from the null space than the others, and that modestly amplifying weaker factors improves the covariance matrix estimation. The PER steps can be viewed as sequential linear shrinkage on eigenvalues, boosting smaller values and contracting larger ones. We also demonstrate that ERSE remains computationally feasible even in high dimensions. Out-of-sample tests across multiple Ken-French datasets reveal that ERSE yields lower portfolio variance than benchmark estimators. Further checks show that ERSE produces covariance matrices with smaller condition numbers, more concentrated and stable portfolio weights, and consistent gains in every subperiod and under different estimation windows.
	
	There are some potential directions for future research. First, extend ERSE to mixed-correlation settings. Second, relax the multivariate normality assumption by deriving shrinkage formulas for nonparametric return distributions. Third, leverage machine-learning techniques on high-dimensional financial data, such as macroeconomic indicators, factor returns, order-book metrics, to learn adaptive shrinkage coefficients and improve ERSE’s responsiveness to changing market conditions.

	\section*{Declaration of Interest}
	We declare that we have no financial and personal relationships with other people or organizations that can inappropriately influence our work, there is no professional or other personal interest of any nature or kind in any product, service and/or company that could be construed as influencing the position presented in, or the review of, the manuscript entitled.
	
	\bibliography{mybib.bib}

\begin{thebibliography}{}

\bibitem[Abadir et~al., 2014]{Abadir2014}
Abadir, K.~M., Distaso, W., and Žikeš, F. (2014).
\newblock Design-free estimation of variance matrices.
\newblock {\em Journal of Econometrics}, 181(2):165--180.

\bibitem[Barberis et~al., 2005]{BARBERIS2005}
Barberis, N., Shleifer, A., and Wurgler, J. (2005).
\newblock Comovement.
\newblock {\em Journal of Financial Economics}, 75(2):283--317.

\bibitem[Barroso and Saxena, 2022]{Barroso2022}
Barroso, P. and Saxena, K. (2022).
\newblock Lest we forget: Learn from out-of-sample forecast errors when
  optimizing portfolios.
\newblock {\em The Review of Financial Studies}, 35(3):1222--1278.

\bibitem[Blau et~al., 2023]{BLAU2023}
Blau, B.~M., Griffith, T.~G., and Whitby, R.~J. (2023).
\newblock Industry regulation and the comovement of stock returns.
\newblock {\em Journal of Empirical Finance}, 73:206--219.

\bibitem[Bodnar et~al., 2017]{Bodnar2017}
Bodnar, T., Mazur, S., and Okhrin, Y. (2017).
\newblock Bayesian estimation of the global minimum variance portfolio.
\newblock {\em European Journal of Operational Research}, 256(1):292--307.

\bibitem[Bodnar et~al., 2018]{Bodnar2018}
Bodnar, T., Parolya, N., and Schmid, W. (2018).
\newblock Estimation of the global minimum variance portfolio in high
  dimensions.
\newblock {\em European Journal of Operational Research}, 266(1):371--390.

\bibitem[Cai et~al., 2020]{Cai2020}
Cai, T.~T., Hu, J., Li, Y., and Zheng, X. (2020).
\newblock High-dimensional minimum variance portfolio estimation based on
  high-frequency data.
\newblock {\em Journal of Econometrics}, 214(2):482--494.

\bibitem[Conlon et~al., 2025]{Conlon2025}
Conlon, T., Cotter, J., and Kynigakis, I. (2025).
\newblock Asset allocation with factor-based covariance matrices.
\newblock {\em European Journal of Operational Research}, 325(1):189--203.

\bibitem[Cui et~al., 2024]{Cui2024}
Cui, L., Hong, Y., Li, Y., and Wang, J. (2024).
\newblock A regularized high-dimensional positive definite covariance estimator
  with high-frequency data.
\newblock {\em Management Science}, 70(10):7242--7264.

\bibitem[Dai et~al., 2019]{DAI2019}
Dai, C., Lu, K., and Xiu, D. (2019).
\newblock Knowing factors or factor loadings, or neither? {Evaluating}
  estimators of large covariance matrices with noisy and asynchronous data.
\newblock {\em Journal of Econometrics}, 208(1):43--79.

\bibitem[Dai et~al., 2024]{Dai2024}
Dai, R., Uematsu, Y., and Matsuda, Y. (2024).
\newblock Estimation of large covariance matrices with mixed factor structures.
\newblock {\em The Econometrics Journal}, 27(1):62--83.

\bibitem[DeMiguel et~al., 2009a]{DeMiguel2009}
DeMiguel, V., Garlappi, L., Nogales, F.~J., and Uppal, R. (2009a).
\newblock A generalized approach to portfolio optimization: Improving
  performance by constraining portfolio norms.
\newblock {\em Management Science}, 55(5):798--812.

\bibitem[DeMiguel et~al., 2009b]{DeMiguel2009a}
DeMiguel, V., Garlappi, L., and Uppal, R. (2009b).
\newblock Optimal versus naive diversification: How inefficient is the {1/N}
  portfolio strategy?
\newblock {\em The Review of Financial Studies}, 22(5):1915--1953.

\bibitem[Ding et~al., 2021]{Ding2021}
Ding, Y., Li, Y., and Zheng, X. (2021).
\newblock High dimensional minimum variance portfolio estimation under
  statistical factor models.
\newblock {\em Journal of Econometrics}, 222(1, Part B):502--515.

\bibitem[Fan et~al., 2008]{Fan2008}
Fan, J., Fan, Y., and Lv, J. (2008).
\newblock High dimensional covariance matrix estimation using a factor model.
\newblock {\em Journal of Econometrics}, 147(1):186--197.

\bibitem[Fan et~al., 2012]{Fan2012}
Fan, J., Li, Y., and Yu, K. (2012).
\newblock Vast volatility matrix estimation using high-frequency data for
  portfolio selection.
\newblock {\em Journal of the American Statistical Association},
  107(497):412--428.

\bibitem[Fan et~al., 2013]{Fan2013}
Fan, J., Liao, Y., and Mincheva, M. (2013).
\newblock Large covariance estimation by thresholding principal orthogonal
  complements.
\newblock {\em Journal of the Royal Statistical Society: Series B (Statistical
  Methodology)}, 75(4):711--730.

\bibitem[Fan et~al., 2018]{Fan2018}
Fan, J., Liu, H., and Wang, W. (2018).
\newblock {Large covariance estimation through elliptical factor models}.
\newblock {\em The Annals of Statistics}, 46(4):1383 -- 1414.

\bibitem[Goto and Xu, 2015]{Goto2015}
Goto, S. and Xu, Y. (2015).
\newblock Improving mean variance optimization through sparse hedging
  restrictions.
\newblock {\em Journal of Financial and Quantitative Analysis},
  50(6):1415--1441.

\bibitem[Green and Hwang, 2009]{Green2009}
Green, T.~C. and Hwang, B.-H. (2009).
\newblock Price-based return comovement.
\newblock {\em Journal of Financial Economics}, 93(1):37--50.

\bibitem[Kim and Fan, 2019]{KIM2019}
Kim, D. and Fan, J. (2019).
\newblock Factor {GARCH-Itô} models for high-frequency data with application
  to large volatility matrix prediction.
\newblock {\em Journal of Econometrics}, 208(2):395--417.

\bibitem[Lam, 2016]{Lam2016}
Lam, C. (2016).
\newblock Nonparametric eigenvalue-regularized precision or covariance matrix
  estimator.
\newblock {\em The Annals of Statistics}, 44(3):928--953, 26.

\bibitem[Ledoit, 1995]{Ledoit1995}
Ledoit, O. (1995).
\newblock {\em Essays on Risk and Return in the Stock Market}.
\newblock PhD thesis, Massachusetts Institute of Technology, Sloan School of
  Management, Cambridge, MA, USA.

\bibitem[Ledoit and Wolf, 2003]{Ledoit2003}
Ledoit, O. and Wolf, M. (2003).
\newblock Improved estimation of the covariance matrix of stock returns with an
  application to portfolio selection.
\newblock {\em Journal of Empirical Finance}, 10(5):603--621.

\bibitem[Ledoit and Wolf, 2004a]{Ledoit2004a}
Ledoit, O. and Wolf, M. (2004a).
\newblock Honey, {I} shrunk the sample covariance matrix.
\newblock {\em Journal of Portfolio Management}, 30(4):110--119.

\bibitem[Ledoit and Wolf, 2004b]{Ledoit2004}
Ledoit, O. and Wolf, M. (2004b).
\newblock A well-conditioned estimator for large-dimensional covariance
  matrices.
\newblock {\em Journal of Multivariate Analysis}, 88(2):365--411.

\bibitem[Ledoit and Wolf, 2008]{Ledoit2008}
Ledoit, O. and Wolf, M. (2008).
\newblock Robust performance hypothesis testing with the sharpe ratio.
\newblock {\em Journal of Empirical Finance}, 15(5):850--859.

\bibitem[Ledoit and Wolf, 2011]{Ledoit2011}
Ledoit, O. and Wolf, M. (2011).
\newblock Robust performances hypothesis testing with the variance.
\newblock {\em Wilmott}, 2011(55):86--89.

\bibitem[Ledoit and Wolf, 2012]{Ledoit2012}
Ledoit, O. and Wolf, M. (2012).
\newblock Nonlinear shrinkage estimation of large-dimensional covariance
  matrices.
\newblock {\em The Annals of Statistics}, 40(2):1024--1060.

\bibitem[Ledoit and Wolf, 2017]{Ledoit2017}
Ledoit, O. and Wolf, M. (2017).
\newblock Nonlinear shrinkage of the covariance matrix for portfolio selection:
  Markowitz meets goldilocks.
\newblock {\em The Review of Financial Studies}, 30(12):4349--4388.

\bibitem[Ledoit and Wolf, 2020]{Ledoit2020}
Ledoit, O. and Wolf, M. (2020).
\newblock Analytical nonlinear shrinkage of large-dimensional covariance
  matrices.
\newblock {\em Annals of Statistics}, 48:3043--3065.

\bibitem[Ledoit and Wolf, 2022a]{Ledoit2022}
Ledoit, O. and Wolf, M. (2022a).
\newblock The power of (non-)linear shrinking: A review and guide to covariance
  matrix estimation.
\newblock {\em Journal of Financial Econometrics}, 20(1):187--218.

\bibitem[Ledoit and Wolf, 2022b]{Ledoit2022a}
Ledoit, O. and Wolf, M. (2022b).
\newblock Quadratic shrinkage for large covariance matrices.
\newblock {\em Bernoulli}, 28(3):1519--1547.

\bibitem[Malceniece et~al., 2019]{MALCENIECE2019}
Malceniece, L., Malcenieks, K., and Putniņš, T.~J. (2019).
\newblock High frequency trading and comovement in financial markets.
\newblock {\em Journal of Financial Economics}, 134(2):381--399.

\bibitem[Markowitz, 1952]{Markowitz1952}
Markowitz, H. (1952).
\newblock Portfolio selection.
\newblock {\em The Journal of Finance}, 7(1):77--91.

\bibitem[M\"orstedt et~al., 2024]{Morstedt2024}
M\"orstedt, T., Lutz, B., and Neumann, D. (2024).
\newblock Cross validation based transfer learning for cross-sectional
  non-linear shrinkage: A data-driven approach in portfolio optimization.
\newblock {\em European Journal of Operational Research}, 318(2):670--685.

\bibitem[Nguyen et~al., 2022]{Nguyen2022}
Nguyen, V.~A., Kuhn, D., and Esfahani, P.~M. (2022).
\newblock Distributionally robust inverse covariance estimation: The
  {Wasserstein} shrinkage estimator.
\newblock {\em Operations Research}, 70(1):490--515.

\bibitem[Pirinsky and Wang, 2006]{PIRINSKY2006}
Pirinsky, C. and Wang, Q. (2006).
\newblock Does corporate headquarters location matter for stock returns?
\newblock {\em The Journal of Finance}, 61(4):1991--2015.

\bibitem[Reh et~al., 2023]{Reh2023}
Reh, L., Krüger, F., and Liesenfeld, R. (2023).
\newblock Predicting the global minimum variance portfolio.
\newblock {\em Journal of Business and Economic Statistics}, 41(2):440--452.

\bibitem[Shi et~al., 2020]{Shi2020}
Shi, F., Shu, L., Yang, A., and He, F. (2020).
\newblock Improving minimum-variance portfolios by alleviating overdispersion
  of eigenvalues.
\newblock {\em Journal of Financial and Quantitative Analysis},
  55(8):2700--2731.

\bibitem[Stein, 1986]{Stein1986}
Stein, C. (1986).
\newblock Lectures on the theory of estimation of many parameters.
\newblock {\em Journal of Soviet Mathematics}, 34(1):1373--1403.

\bibitem[Touloumis, 2015]{Touloumis2015}
Touloumis, A. (2015).
\newblock Nonparametric stein-type shrinkage covariance matrix estimators in
  high-dimensional settings.
\newblock {\em Computational Statistics and Data Analysis}, 83:251--261.

\bibitem[Vaart, 1961]{Vaart1961}
Vaart, H.~R. (1961).
\newblock On certain characteristics of the distribution of the latent roots of
  a symmetric random matrix under general conditions.
\newblock {\em Annals of Mathematical Statistics}, 32:864--873.

\bibitem[Wang et~al., 2021]{Wang2021}
Wang, H., Peng, B., Li, D., and Leng, C. (2021).
\newblock Nonparametric estimation of large covariance matrices with
  conditional sparsity.
\newblock {\em Journal of Econometrics}, 223(1):53--72.

\bibitem[Won et~al., 2013]{Won2013}
Won, J.-H., Lim, J., Kim, S.-J., and Rajaratnam, B. (2013).
\newblock Condition-number-regularized covariance estimation.
\newblock {\em Journal of the Royal Statistical Society Series B: Statistical
  Methodology}, 75(3):427--450.

\end{thebibliography}
		
	\end{document}